\documentclass[11pt]{article}
\usepackage{geometry}

\usepackage[utf8]{inputenc}
\usepackage{amsthm}
\usepackage{amssymb}
\usepackage{enumitem}
\usepackage{mathtools}
\usepackage{authblk}

\usepackage{hyperref}
\usepackage{cleveref}

\usepackage[dvipsnames]{xcolor}
\usepackage{tikz}
\usetikzlibrary{decorations.pathreplacing}

\theoremstyle{plain}
\newtheorem{theo}{Theorem}
\crefname{theo}{Theorem}{Theorems}
\newtheorem{lem}[theo]{Lemma}
\crefname{lem}{Lemma}{Lemmata}

\newtheorem{cor}[theo]{Corollary}
\crefname{cor}{Corollary}{Corollarys}

\theoremstyle{definition}
\newtheorem{defn}[theo]{Definition}

\theoremstyle{remark}
\newtheorem{rem}[theo]{Remark}

\newenvironment{canonicallabeling}{\vspace{0.1cm}\noindent(\textit{Canonical
labeling.})}{}
\newenvironment{correctness}{\vspace{0.1cm}\noindent(\textit{Correctness.})}{}

\newenvironment{runningtime}{\vspace{0.1cm}\noindent(\textit{Running time.})}{}

\newcommand{\case}[1]{\item[\itshape\mdseries Case #1:]}
\newenvironment{cs}{\begin{description}}{\end{description}}

\newcommand{\step}[1]{\item[\itshape\mdseries Step #1:]}
\newenvironment{stp}{\begin{description}}{\end{description}}

\renewcommand{\mathbf}[1]{\textit{\bfseries #1}}

\renewcommand{\tilde}{\widetilde}
\renewcommand{\hat}{\widehat}
\renewcommand{\bar}{\overline}

\newcommand{\lcosetgen}{\bar{\langle}}
\newcommand{\rcosetgen}{\bar{\rangle}}

\renewcommand{\phi}{\varphi}
\renewcommand{\epsilon}{\varepsilon}

\newcommand{\NN}{{\mathbb N}}

\newcommand{\CA}{{\mathcal A}}

\newcommand{\CF}{{\mathcal F}}
\newcommand{\Cf}{{\mathfrak f}}

\newcommand{\CI}{{\mathcal I}}

\newcommand{\CO}{{\mathcal O}}

\newcommand{\CS}{{\mathcal S}}

\newcommand{\CX}{{\mathcal X}}

\newcommand{\sym}{\operatorname{Sym}}
\newcommand{\aut}{\operatorname{Aut}}
\newcommand{\iso}{\operatorname{Iso}}
\newcommand{\id}{\operatorname{id}}
\newcommand{\tw}{\operatorname{tw}}

\newcommand{\gammak}{\tilde\Gamma_{k+1}}
\newcommand{\Gammak}{\Gamma_{k+1}}

\newcommand{\can}{\operatorname{CL}}
\newcommand{\lab}{\operatorname{Label}}

\newcommand{\cf}{\operatorname{CF}}

\newcommand{\depth}{\operatorname{depth}}
\newcommand{\cupdot}{\mathbin{\mathaccent\cdot\cup}}
\newcommand{\nin}{\notin}

\DeclareMathOperator{\Pow}{Pow}

\DeclareMathOperator{\cw}{cw}

\DeclareMathOperator{\rg}{rg}
\DeclareMathOperator{\dist}{dist}
\DeclareMathOperator{\fdeg}{fdeg}
\DeclareMathOperator{\last}{last}
\DeclareMathOperator{\parent}{par}

\newcommand{\polylog}[1]{\operatorname{polylog}(#1)}
\newcommand{\poly}[1]{\operatorname{poly}(#1)}

\newcommand{\funcsmall}{c_{\textrm{S}}} 
\newcommand{\funcmedium}{c_{\textrm{M}}} 
\newcommand{\funclarge}{c_{\textrm{L}}} 

\begin{document}

\title{An improved isomorphism test for bounded-tree-width graphs}

\author[$\dagger$]{Martin Grohe}
\author[$\dagger$]{Daniel Neuen}
\author[$\ddagger$]{Pascal Schweitzer}
\author[$\dagger$]{Daniel Wiebking}

\affil[$\dagger$]{RWTH Aachen University\\\texttt{\normalsize{\{grohe,neuen,wiebking\}@informatik.rwth-aachen.de}}}
\affil[$\ddagger$]{University of Kaiserslautern\\\texttt{\normalsize{schweitzer@cs.uni-kl.de}}}
\maketitle

\begin{abstract}
  We give a new fpt algorithm testing isomorphism of $n$-vertex
  graphs of tree width $k$ in time $2^{k\polylog k}\poly n$,
  improving the fpt algorithm due to Lokshtanov, Pilipczuk, Pilipczuk,
  and Saurabh (FOCS 2014), which runs in time
  $2^{\mathcal{O}(k^5\log k)}\poly n$. Based on an improved version
  of the isomorphism-invariant graph decomposition technique introduced by
  Lokshtanov et al., we prove restrictions on the structure of the automorphism
  groups of graphs of tree width $k$. Our
  algorithm then makes heavy use of the group theoretic techniques
  introduced by Luks (JCSS 1982) in his isomorphism test for bounded
  degree graphs and Babai (STOC 2016) in his quasipolynomial
  isomorphism test. In fact, we even use Babai's algorithm as a black
  box in one place.

  We also give a second algorithm which, at the price of a slightly worse
  running time $2^{\mathcal{O}(k^2 \log k)}\poly n$, avoids the use of Babai's
  algorithm and, more importantly, has the additional benefit that it can also used as a
  canonization algorithm.
\end{abstract}

\section{Introduction}

Already early on in the beginning of research on the graph isomorphism problem (which asks for structural equivalence of two given input graphs)
a close connection to the structure and study of the automorphism group of a graph was observed.
For example, Mathon \cite{DBLP:journals/ipl/Mathon79} argued that the isomorphism problem is polynomially equivalent to the task of computing a generating set for the automorphism group and also to computing the size of the automorphism group.

With Luks's polynomial time isomorphism test
for graphs of bounded degree~\cite{luks1982isomorphism}, the striking usefulness
of group theoretic techniques for isomorphism problems became apparent and they have been exploited ever since (e.g.~\cite{babai1983canonical,miller1983isomorphism,Ponomarenko,Neuen16}). 
In his algorithm, Luks shows and uses that the automorphism group of a graph of bounded degree, after a vertex has been fixed, has a very restricted structure.
More precisely, the group is in the class $\Gamma_k$ of all groups whose non-Abelian composition factors are isomorphic to a subgroup of the symmetric group $\sym(k)$.

Most recently, Babai's quasi-polynomial time algorithm for general graph isomorphism \cite{DBLP:conf/stoc/Babai16} adds several novel techniques to tame and manage the groups that may appear as the automorphism group of the input graphs.

A second approach towards isomorphism testing is via decomposition techniques (e.g.\ \cite{DBLP:journals/jal/Bodlaender90,gromar15,DBLP:conf/focs/GroheS15}).
These decompose the graph into smaller pieces while maintaining control of the complexity of the interplay between the pieces.
When taking this route it is imperative to decompose the graph in an isomorphism-invariant fashion so as not to compare two graphs that have been decomposed in structurally different ways.

A prime example of this strategy is 
Bodlaender's isomorphism test~\cite{DBLP:journals/jal/Bodlaender90} for graphs of bounded treewidth. 
Bodlaender's algorithm is a dynamic programming algorithm that takes into
account all~$k$-tuples of vertices that separate the graph, leading to a running
time of~$\CO(n^{k+c})$ to test isomorphism of graphs of tree width at most~$k$.

Only recently,  Lokshtanov, Pilipczuk, Pilipczuk, and
Saurabh~\cite{lokshtanov2017fixed} designed a fixed-parameter tractable
isomorphism test for graphs of bounded tree width which has a running time of $2^{\CO(k^5\log k)}\operatorname{poly(n)}$.
This algorithm first ``improves'' a given input graph~$G$ to a graph~$G^k$ by
adding an edge between every pair of vertices between which more
than~$k$-internally vertex disjoint paths exist.
The improved graph~$G^k$ isomorphism-invariantly decomposes along clique separators  into clique-separator free parts, which we will call \emph{basic} throughout the paper.
The decomposition can in fact be extended to an isomorphism-invariant tree decomposition into basic parts, as was shown in~\cite{DBLP:conf/stacs/ElberfeldS16} to design a logspace isomorphism test for graphs of bounded tree width.
For the basic parts, Lokshtanov et al.~\cite{lokshtanov2017fixed} show that,
after fixing a vertex of sufficiently low degree, is it possible to compute an
isomorphism-invariant tree decomposition whose bags have a size at most
exponential in~$k$ and whose adhesion is at most~$\CO(k^3)$. 
They use this invariant decomposition to compute a canonical form essentially by a brute-force dynamic programming algorithm. 

The problem of computing a canonical form is the task to compute, to a given input graph~$G$, a graph~$G'$ isomorphic to~$G$ such that the output~$G'$ depends only on the isomorphism class of~$G$ and not on~$G$ itself.

The isomorphism problem reduces to the task of computing a
canonical form: for two given input graphs we compute their canonical forms and check whether the canonical forms are equal (rather than isomorphic).

As far as we know, computing a canonical form could be algorithmically more difficult than testing isomorphism.
It is usually not very difficult to turn combinatorial isomorphism tests into canonization algorithms, sometimes the algorithms are canonization algorithms in the first place.
However, canonization based on group theoretic isomorphism tests is more challenging. For example, it is still open whether there is a graph canonization algorithm running in quasipolynomial time.

\subsection*{Our Results} 

Our main result is a new fpt algorithm testing isomorphism of graphs of bounded tree width.

\begin{theo}
  There is a graph isomorphism test running in time
  \[
    2^{k\polylog k}\poly n,
  \]
  where $n$ is the size and $k$ the minimum tree width of the input graphs.
\end{theo}

In the first part of the paper, we analyze the structure of the automorphism groups of a graph~$G$ of tree width $k$. 
Following \cite{lokshtanov2017fixed} and \cite{DBLP:conf/stacs/ElberfeldS16}, we pursue a two-stage decomposition strategy for graphs of bounded tree width, where in the first step we decompose the improved graph along clique separators into basic parts.
We observe that these basic parts are essential for understanding the automorphism groups. 
We show (Theorem~\ref{theo:labelDecomposition}) that with respect to a fixed
vertex~$v$ of degree at most~$k$, we can construct for each basic graph~$H$ an
isomorphism-invariant tree decomposition of width at most $2^{\CO(k\log k)}$ and
adhesion at most $\CO(k^2)$ where, in addition, each bag is equipped with a
graph of small degree which is defined in an isomorphism-invariant way and gives us insight about the structure of the bag.
In particular, using Luks results \cite{luks1982isomorphism}, this also restricts the structure of the automorphism group.

Our construction is based on a similar construction of an isomorphism-invariant
tree decomposition in \cite{lokshtanov2017fixed}. Compared to that construction,
we improve the adhesion (that is, the maximum size of intersections between
adjacent bags of the decomposition) from $\CO(k^3)$ to $\CO(k^2)$. More
importantly, we expand the decomposition by assigning a group and a graph to each bag.

Using these groups, we can prove that $\aut(H)_v$ (the group of all automorphisms of $H$ that keep the vertex $v$ fixed) is a~$\Gammak$ group. 
This significantly restricts possible automorphism groups.
Moreover, using the graph structure assigned to each bag, we can also compute the automorphism group of a graph of tree width $k$ within the desired time bounds.
The first, already nontrivial step towards computing the automorphism group, is a reduction from arbitrary graphs of tree width $k$ to basic graphs.
The second step reduces the problem of computing the automorphism group of a basic graph to the
problem of computing the automorphism group of a structure that we call an \emph{expanded $d$-ary tree}.
In the reduction, the parameter $d$ will be polynomially bounded in $k$.
Then as the third step, we can apply a recent result \cite{bounded-degree-small-diameter} due to the first three authors
that allows us to compute the automorphism groups of such expanded $d$-ary trees.
This result is heavily based on techniques introduced by Babai~\cite{DBLP:conf/stoc/Babai16} in his quasipolynomial isomorphism test.
In fact, it even uses Babai's algorithm as a black box in one place.

We prove a second result that avoids the results of \cite{bounded-degree-small-diameter, DBLP:conf/stoc/Babai16} and even allows us to compute canonical forms, albeit at the price of an increased running time. 

\begin{theo}
  There is a graph canonization algorithm running in time
  \[
    2^{\mathcal{O}(k^2\log k)}\poly n,
  \]
  where $n$ is the size and $k$ the tree width of the input graph.
\end{theo}

Even though it does not employ Babai's new techniques, this algorithm still heavily depends on the group theoretic machinery.
As argued above, the design of group theoretic canonization algorithms often requires extra work, and can be slightly technical, compared to the design of an isomorphism algorithm.
Here, we need to combine the group theoretic canonization techniques going back to Babai and Luks~\cite{babai1983canonical} with graph decomposition techniques, which poses additional technical challenges and requires new canonization subroutines.

\subsection*{Organization of the paper}

In the next section we introduce the necessary preliminaries.
The next two Sections~\ref{sec:clique:seps:and:improved:graphs}--\ref{sec:basic} of the paper
deal with the decomposition of bounded tree width graphs.
They describe the
isomorphism-invariant decomposition into basic parts
and the isomorphism-invariant decomposition of the basic parts with respect to a
fixed vertex of low degree.

The two following Sections~\ref{sec:iso:sub}--\ref{sec:iso:alg}
are concerned with isomorphism.
We define a particular subproblem, coset-hypergraph-isomorphism,
identified to be of importance.
Then, we assemble the recursive isomorphism algorithm.

In the last Section~\ref{sec:canoniztion:tools} we
devise several subroutines for
canonization in general and assemble the recursive canonization algorithm.


\section{Preliminaries}

\paragraph{Graphs}

We use standard graph notation.
All graphs $G=(V,E)$ considered are undirected finite simple graphs.
We denote an edge $\{u,v\}\in E$ by $uv$.
Let $U,W\subseteq V$ be subsets of vertices.
We write $E(U,W)$ for the edges
with one vertex in $U$ and the other vertex from $W$,
whereas $E(U)$ are the edges with both vertices in $U$.
By $N(U)$, we denote the neighborhood of $U$, i.e,
all vertices outside $U$ that are adjacent to $U$.
For the induced subgraph on $U$, we write $G[U]$,
whereas $G-U$ is the induced subgraph on $V\setminus U$.
A \emph{rooted graph} is triple $G=(V,E,r)$ where $r \in V$ is the \emph{root} of the graph.
For two vertices $v,w \in V$ we denote by $\dist_G(v,w)$ the distance between $v$ and $w$, i.e.\ the length of the shortest path from $v$ to $w$.
The \emph{depth} of a rooted graph is the maximum distance from a vertex to the root, that is, $\depth(G) = \max_{v\in V} \dist_G(r,v)$.
The \emph{forward-degree} of a vertex $v \in V$ is $\fdeg(v) = |\{w \in N(v) \mid \dist(w,r) = \dist(v,r)+1\}|$.
Note that $|V| \leq (d+1)^{\depth(G)}$ where $d = \max_{v \in V} \fdeg(v)$ is the maximal forward-degree.

\paragraph{Separators}

A pair $(A,B)$ where $A \cup B = V (G)$ is called a \emph{separation} if
$E(A\setminus B,B\setminus A) = \emptyset$. In this case we call
$A\cap B$ a \emph{separator}.
A separation $(A,B)$ is
a \emph{$(L,R)$-separation} if $L\subseteq A$ and $R\subseteq B$
and in this case $A\cap B$ is called a \emph{$(L,R)$-separator}.
A separation $(A,B)$ is a called a \emph{clique separation}
if $A\cap B$ is a clique and $A\setminus B\neq\emptyset$ and $B\setminus
A\neq\emptyset$.
In this case we call
$A\cap B$ a \emph{clique separator}.

\paragraph{Tree Decompositions}
\begin{defn}
A \emph{tree decomposition} of a graph $G$ is a pair $(T,\beta)$, where
$T$ is a \underline{rooted} tree and $\beta: V(T)\to \Pow(V (G))$
is a mapping into the power set of~$V(G)$ such that:
\begin{enumerate}
\item for each vertex $v\in V (G)$, the set
$\{t \in V (G)~|~ v \in \beta(t)\}$ induces a
nonempty and connected subtree of T, and
\item for each edge $e\in E(G)$, there exists $t\in V(T)$
such that $e\subseteq\beta(t)$.
\end{enumerate}
\end{defn}
Sets $\beta(t)$ for $t\in V (T)$ are called the \emph{bags}
of the
decomposition, while sets $\beta(s)\cap\beta(t)$
for $st \in E(T)$ are called
the \emph{adhesions sets}.
The \emph{width}
of a tree decomposition $T$ is equal to its maximum
bag size decremented by one, i.e.
$\max_{t\in V (T)} |\beta(t)| - 1$.
The \emph{adhesion width} of $T$ is equal to its maximum
adhesion size, i.e., $\max_{st\in E(T)} |\beta(s) \cap\beta(t)|$.
The \emph{tree width} of a graph,
denoted by $\tw(G)$, is equal to the minimum width of
its tree decompositions.

A graph $G$ is \emph{$k$-degenerate} if every subgraph of
$G$ has a vertex with degree at most $k$.
It is well known that every graph of tree width $k$
is $k$-degenerate.

\paragraph{Groups}
For a function $\phi \colon V \rightarrow V'$ and $v\in V$ we write $v^{\phi}$ for
the image of $v$ under $\phi$, that is, $v^{\phi} = \phi(v)$.
We write composition of functions from \underline{left to right}, e.g,
$v^{(\sigma\rho)}=(v^\sigma)^\rho=\rho(\sigma(v))$.
By $[t]$ we denote the set of natural number from 1
to $t$.
By $\sym(V)$ we denote the symmetric group on a
set $V$ and we also write $\sym(t)$ for $\sym([t])$.
We use upper Greek letters $\Delta,\Phi,\Gamma,\Theta$ and $\Psi$ for
permutation groups.

\paragraph{Labeling cosets}
A \emph{labeling coset} of a set $V$ 
is a set of bijective mappings
$\tau\Delta$
where $\tau$ is a
bijection from~$V$ to~$[|V|]$ and
and $\Delta$ is a subgroup of $\sym(|V|)$.
By $\lab(V)$, we denoted
the labeling coset $\tau\sym(|V|)$.
We say that $\tau\Delta$ is a \emph{labeling subcoset}
of a labeling coset $\rho\Theta$, written $\tau\Delta\leq\rho\Theta$,
if $\tau\Delta$ is a subset of $\rho\Theta$
and $\tau\Delta$ forms a labeling coset again.
Sometimes we will choose a single symbol to denote a labeling coset $\tau\Delta$. For this will usually use the upper greek letter $\Lambda$.
Recall that $\Gamma_k$ denotes the class of all
finite groups whose non-Abelian composition
factors are isomorphic to subgroups of $\sym(k)$.
Let $\tilde\Gamma_k$ be the class of all labeling cosets
$\Lambda=\tau\Delta$ such that $\Delta\in\Gamma_k$.

\paragraph{Orderings on sets of natural numbers}
We extend the natural ordering of the natural numbers to finite sets of natural
numbers. For two such sets~$M_1,M_2$ we define~$M_1\prec M_2$ if~$|M_1|<|M_2|$
or if $|M_1| = |M_2|$ and the smallest element of~$M_1\setminus M_2$ is smaller than the smallest element of~$M_2\setminus M_1$.

\paragraph{Isomorphisms}
In this paper we will always define what the isomorphism
between our considered objects are.
But this can also be done in a more general context.
Let $\phi:V\to V'$.
For a vector $(v_1,\ldots,v_k)$ we define $(v_1,\ldots,v_k)^\phi$
as $(v_1^\phi,\ldots,v_k^\phi)$ inductively.
Analogously, for a set we define $\{v_1,\ldots,v_n\}^\phi$ as
$\{v_1^\phi,\ldots,v_n^\phi\}$.
For a labeling coset $\Lambda\leq\lab(V)$
we write $\Lambda^\phi$ for $\phi^{-1}\Lambda$.
In the paper we will introduce isomorphisms $\iso(X,X')$
for various objects $X$ and $X'$.
Unless otherwise stated these are all $\phi:V\to V'$
such that $X^\phi=X'$ where we apply $\phi$ as previously defined.
For example, the isomorphism between two graphs $G$ and $G'$
are all $\phi:V\to V'$ such that
$G^\phi=G'$ which means that $G$ has an edge $uv$,
if and only if $G'$ has the edge $u^\phi v^\phi$.


\section{Clique separators and improved graphs}\label{sec:clique:seps:and:improved:graphs}

To perform isomorphism tests of graphs of bounded tree width, a crucial step in~\cite{lokshtanov2017fixed} is to deal with clique separators.
For this step the concept of a~$k$-improved graph is the key.

\begin{defn}[\cite{lokshtanov2017fixed}]
The \emph{$k$-improvement} of a graph $G$ is the graph $G^k$ obtained from $G$
by connecting every pair of non-adjacent
vertices $v,w$ for which there are more than $k$ pairwise internally vertex disjoint
paths connecting~$v$ and~$w$.
We say that a graph $G$ \emph{is $k$-improved}
when $G^k=G$.

A graph is \emph{$k$-basic} if it is $k$-improved and
does not have any separating cliques.
In particular, a $k$-basic graph is connected.
\end{defn}

We summarize several structural properties of~$G^k$.

\begin{lem}[\cite{lokshtanov2017fixed}]\label{lem:imp}
Let $G$ be a graph and $k\in\NN$.
\begin{enumerate}
\item The $k$-improvement $G^k$ is $k$-improved,
i.e., $(G^k)^k=G^k$.
\item Every tree
decomposition $(T,\beta)$ of $G$ of width at most $k$
is also a tree
decomposition of $G^k$.
\item There exists an algorithm that,
given $G$ and $k$,
runs in $\CO(k^2n^3)$ time and
either correctly concludes that $\tw(G) > k$, or
computes~$G^k$.
\end{enumerate}
\end{lem}

Since the construction of~$G^k$ from~$G$ is isomorphism-invariant, the concept of the improved graph can be exploited for isomorphism testing and canonization.
A~$k$-basic graph has severe limitations
concerning its structure as we explore in the following sections. In the
canonization algorithm from~\cite{lokshtanov2017fixed} a result of
Leimer~\cite{DBLP:journals/dm/Leimer93}
is exploited that says that every graph has a tree decomposition into clique-separator free parts, and the family of bags is isomorphism-invariant.
While it is usually sufficient to work with an isomorphism-invariant set of bags (see \cite{DBLP:conf/swat/OtachiS14}) we 
actually require an isomorphism invariant decomposition, which can indeed be obtained.

\begin{theo}[\cite{DBLP:journals/dm/Leimer93},\cite{DBLP:conf/stacs/ElberfeldS16}]
\label{theo:cliqueDecomposition}
There is an algorithm that, given a connected graph~$G$,
computes a
tree decomposition $(T,\beta)$ of $G$, 
called clique separator decomposition,
with the following properties.

\begin{enumerate}
\item For every $t\in V(T)$ the graph $G[\beta(t)]$
is clique-separator free (and in particular connected).
\item Each adhesion set of $(T,\beta)$ is a clique.
\item $|V(T)|\in \CO(|V(G)|)$.
\item For each bag~$\beta(t)$ the adhesion sets of the children are all equal
to~$\beta(t)$ or the adhesion sets of the children are all distinct.
\end{enumerate}
The algorithm runs in polynomial time and the output of the algorithm is isomorphism-invariant (w.r.t.~$G$).
\end{theo}


\section{Decomposing basic graphs}\label{sec:basic}

In this section, we shall construct bounded-width tree decompositions
of $k$-basic graphs of tree width at most $k$. Crucially, these
decompositions will be isomorphism invariant after fixing one vertex
of the graph. Our construction refines a similar construction of
\cite{lokshtanov2017fixed}.

Let us define three parameters~$\funcsmall$,~$\funcmedium$, and~$\funclarge$ (small, medium and large) that depend on~$k$ as follows:
\begin{align*}
\funcsmall\coloneqq 6(k+1),&&
\funcmedium\coloneqq\funcsmall+\funcsmall(k+1),&&
\funclarge\coloneqq\funcmedium+2(k+1)\binom{\funcmedium}{k+2}^2.
\end{align*}

Note that $\funcsmall \in \CO(k)$, $\funcmedium \in \CO(k^{2})$ and $\funclarge \in 2^{\CO(k\log k)}$.
The interpretation of these parameters is that~$\funcmedium$ will bound the
size of the adhesion sets and~$\funclarge$ will bound the bag size. The parameter~$\funcsmall$ is used by the algorithm which in certain situations behaves differently depending on sets being larger than~$\funcsmall$ or not.

The bound $\funcmedium\in\CO(k^2)$  improves the corresponding bound
$\funcmedium\in\CO(k^3)$  in~\cite{lokshtanov2017fixed}. However, the
more significant extension of the construction
in~\cite{lokshtanov2017fixed} is that in addition to the tree
decomposition we also construct both an isomorphism-invariant
graph of bounded forward-degree and depth
and an isomorphism-invariant~$\Gammak$-group associated with each bag.

The \emph{weight} of a set $S\subseteq V(G)$
with respect to a (weight) function
$w:V(G)\to\NN$ is $\sum_{v\in S}w(v)$.
The \emph{weight} of a separation $(A,B)$ is the weight of
its separator $A\cap B$.
For sets~$L,R\subseteq V(G)$, among all $(L,R)$-separations~$(A,B)$ of
minimal weight there exists a unique separation with an inclusion
minimal $A$. For this separation we call~$A\cap B$
the \emph{leftmost minimal separator}
and denote it by $S_{L,R}(w)$.
Moreover, we define $S_{L,R} = S_{L,R}(\textbf{1})$ where $\textbf{1}$ denotes the function that maps every vertex to $1$.

For $U\subseteq V(G)$ we define a
weight function $w_{U,k}$ such that
$w_{U,k}(u)=k$ for all $u\in U$ and
$w_{U,k}(v)=1$ for all $v\in V\setminus U$.
Given a weight function~$w$, using Menger's theorem and the
Ford-Fulkerson algorithm it is possible to compute $S_{L,R}(w)$.
The following lemma generalizes Lemma 3.2 of
\cite{lokshtanov2017fixed}. Through this generalization we obtain the
adhesion bound $\CO(k^2)$ for our decomposition.

\begin{lem}\label{lem:union:of:weighted:minimal:separations}
Let $G$ be a graph, let $S \subseteq V (G)$ be a subset of vertices, and let
$\{T_i\subseteq V(G)\}_{i\in[t]}$ and
$\{w_i:V(G)\to\NN\}_{i\in[t]}$ 
be families where each $T_i$ is a minimum weight
$(L_i,R_i)$-separator with respect to $w_i$ for some $L_i,R_i\subseteq S$.
Let $w:V(G)\to\NN$ be another weight function such that
for all $i\in[t]$:
\begin{enumerate}
\item\label{lem:sub1} $w(v)=w_i(v)$ for all $v\in V(G)\setminus S$, and
\item\label{lem:sub2} $w(v)\geq w_i(v)$ for all $v\in V(G)$.
\end{enumerate}
Let $D \coloneqq S \cup\bigcup_{i\in[t]} T_i$.
Suppose that $Z$ is the vertex set of any connected component of $G-D$.
Then $w(N(Z))\leq w(S)$.
\end{lem}

\begin{proof}
We proceed by induction on $t$. For $t=0$ we have $N(Z) \subseteq D = S$, so
the claim is trivial.
Assume $t\geq 1$ and define $D'\coloneqq S \cup\bigcup_{i\in[t-1]} T_i$.
Let $Z$ be
the vertex set of a connected component of $G- D$, and let
$Z'\supseteq Z$ be the vertex set of the connected component of $G-
D'$ containing~$Z$.
Let $L_t,R_t\subseteq S$ be sets such that $T_t$ is a minimum weight
$(L_t,R_t)$-separator with respect to $w_t$. 
Let $(A,B)$ be an $(L_t,R_t)$-separation with 
separator $A\cap B=T_t$. Without loss of generality we may assume that
$Z\subseteq A\setminus B$.
The three sets~$A\setminus T_t$,~$T_t$, and~$B\setminus T_t$
partition~$V(G)$. Similarly the three sets~$V\setminus (Z\cup
N(Z'))$,~$N(Z')$ and $Z'$ partition~$V(G)$.
We define~$Q_{i,j}$ to be the intersection of the~$i$-th set of the
first triple with the~$j$-th set of the second triple. This way
the sets~$Q_{i,j}$ with~$i,j\in \{1,2,3\}$ partition $V(G)$ into 9 parts
as shown in Figure~\ref{fig:decomposition:into:9:sets}.

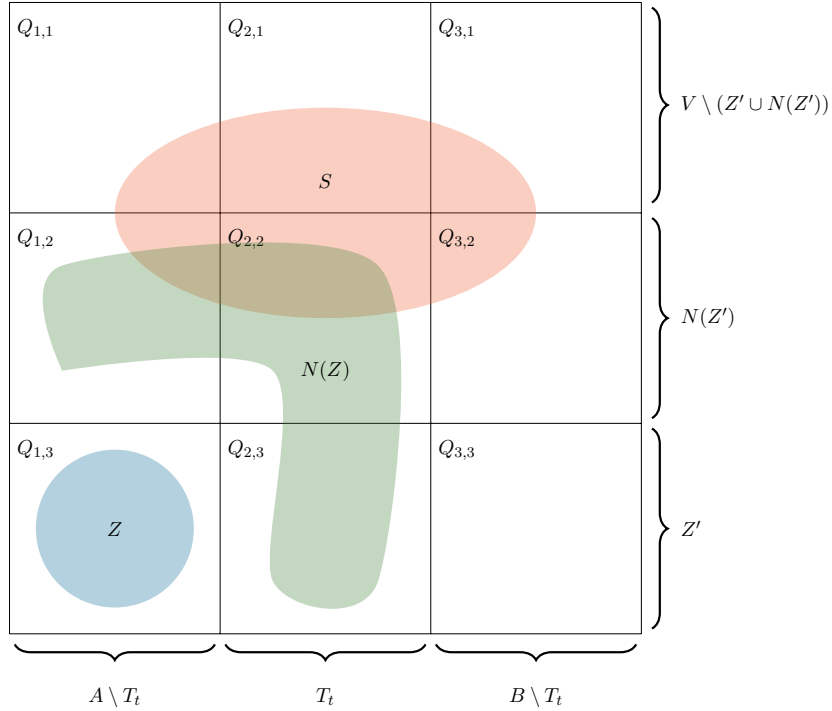
\begin{figure}[h]
\centering{
\scalebox{0.7}{
\begin{tikzpicture}

\draw[step=4.0] (0,0) grid (12,12);
\fill[RedOrange,opacity=0.3] (6,8) ellipse (4cm and 2cm)
node[yshift=0.6cm,black,opacity=1] {$S$};

\foreach \i in {1,2,3}{
   \foreach \j in {1,2,3}{
\draw (4*\i-3.5,-4*\j+15.5) node {$Q_{\i,\j}$};
    }
}

\fill[MidnightBlue,opacity=0.3] (2,2) ellipse (1.5cm and 1.5cm)
node[black,opacity=1] {$Z$};

\fill[OliveGreen,opacity=0.3] plot [smooth, tension=0.5]
coordinates { (1,5) (5,5)
(5,1) (7,1) (7,7) (1,7) (1,5)};
\draw (6,5) node {$N(Z)$};

\draw [very thick,decorate,decoration={brace,amplitude=0.3cm}]
(12.2,11.9) -- (12.2,8.2) node [black,midway,xshift=.4cm,anchor=west]
{$V\setminus (Z'\cup N(Z'))$};
\draw [very thick,decorate,decoration={brace,amplitude=0.3cm}]
(12.2,7.9) -- (12.2,4.1) node [black,midway,xshift=.4cm,anchor=west]
{$N(Z')$};
\draw [very thick,decorate,decoration={brace,amplitude=0.3cm}]
(12.2,3.9) -- (12.2,0.1) node [black,midway,xshift=.4cm,anchor=west]
{$Z'$};

\draw [very thick,decorate,decoration={brace,amplitude=0.3cm}]
(11.9,-0.2) -- (8.1,-0.2) node [black,midway,yshift=-1cm]
{$B \setminus T_t$};
\draw [very thick,decorate,decoration={brace,amplitude=0.3cm}]
(7.9,-0.2) -- (4.1,-0.2) node [black,midway,yshift=-1cm]
{$T_t$};
\draw [very thick,decorate,decoration={brace,amplitude=0.3cm}]
(3.9,-0.2) -- (0.1,-0.2) node [black,midway,yshift=-1cm]
{$A \setminus T_t$};

\end{tikzpicture}
}
}

\caption{Sets appearing in the proof of~\cref{lem:union:of:weighted:minimal:separations}\label{fig:decomposition:into:9:sets}.}
\end{figure}
We have $w(N(Z'))\le w(S)$ by the induction
hypothesis.
Since $N(Z')=Q_{1,2} \cup Q_{2,2} \cup Q_{3,2}$ and $N(Z) \subseteq  Q_{1,2}
\cup Q_{2,2} \cup Q_{2,3}$,
it suffices to show $w(Q_{2,3}) \leq w(Q_{3,2})$.
Observe that
$Q_{2,1} \cup Q_{2,2} \cup Q_{3,2}$ is also an $(L_t,R_t)$-separator,
because $L_t\subseteq A$ and  $R_t\subseteq B$ by the choice of
$(A,B)$, and $R_t\subseteq S\subseteq V(G)\setminus Z'$.

By the minimality of $T_t$, we have $w_t(T_t)\le w_t(Q_{2,1} \cup
Q_{2,2} \cup Q_{3,2})$, and as $T_t=Q_{2,1} \cup Q_{2,2} \cup
Q_{2,3}$, this implies $w_t(Q_{2,3}) \leq w_t(Q_{3,2})$.
By Assumptions (1) and (2) of the lemma, it follows that
\[w(Q_{2,3})  \overset{\mathclap{\ref{lem:sub1}}}{=}w_t(Q_{2,3})\leq w_t(Q_{3,2}) \overset{\mathclap{\ref{lem:sub2}}}{\leq}   w(Q_{3,2})  .\qedhere\]
\end{proof}

The lemma can be used to extend a set of vertices~$S$ that is not a
clique separator to a set~$D$ in an isomorphism-invariant fashion
while controlling the size of the adhesions sets of the components
of~$G-D$.
It will be important for us that we can also extend a labeling coset of~$S$ to
a labeling coset of~$D$ and furthermore
construct a graph of bounded forward-degree and depth
associated with $D$ and $S$.

\begin{lem}\label{lem:B}
Let $k\in\NN$ and let $G$ be a graph that is
$k$-improved.
Let $S\subseteq V (G)$ and let
$\Lambda\leq\lab(S)$ be a labeling coset such that\\[-1.5ex]
\begin{minipage}[t]{0.45\textwidth}
\begin{enumerate}
\item $\emptyset\subsetneq S \subsetneq V (G)$,
\item $|S|\leq \funcmedium$,
\item $S$ is not a clique,
\end{enumerate}
\end{minipage}
\begin{minipage}[t]{0.5\textwidth}
\begin{enumerate}
\setcounter{enumi}{3}
\item $G- S$ is connected,
\item $S = N_G(V (G) \setminus S)$, and
\item $\Lambda\in\gammak$.
\end{enumerate}
\end{minipage}\\[1.5ex]
There is
an algorithm that
either correctly concludes that $tw(G) > k$,
or finds a proper superset $D$ of $S$ and a 
labeling coset  $\hat\Lambda\leq\lab(D)$
and a connected rooted graph $H$ with the
following properties:\\[-1.5ex]
\begin{minipage}[t]{0.45\textwidth}
\begin{enumerate}[label=(\Alph*)]
\item\label{a:a} $D \supsetneq S$,
\item\label{a:b} $|D|\leq \funclarge$,
\end{enumerate}
\end{minipage}
\begin{minipage}[t]{0.5\textwidth}
\begin{enumerate}[label=(\Alph*)]
\setcounter{enumi}{2}
\item\label{a:c} if $Z$ is the vertex set of any connected component of $G
- D$, then $|N(Z)| \leq \funcmedium$, 
\item\label{a:d} $\hat\Lambda\in\gammak$,
\end{enumerate}
\end{minipage}
\begin{enumerate}[label=(\Alph*)]
\setcounter{enumi}{4}
\item\label{a:e} $D\subseteq V(H)$, $\depth(H)\leq k+3$
and
$\fdeg(v)\in k^{\CO(1)}$ for
all $v\in V(H)$.
\end{enumerate}
The algorithm runs in time $2^{\CO(k \log k)}|V (G)|^{\CO(1)}$ and the
output~$(D,\hat\Lambda,H)$ is isomorphism-invariant (w.r.t. the input
data~$G,S,\Lambda$ and~$k$).
\end{lem}

Here, the output of an algorithm $\CA$
is isomorphism-invariant
if all isomorphisms between two input data $(G,S,\Lambda,k)$
and $(G',S',\Lambda',k')$
extends to an isomorphism between the output $(D,\hat\Lambda,H)$
and $(D',\hat\Lambda',H')$
(An isomorphism between $(G,S,\Lambda,k)$
and $(G',S',\Lambda',k')$ is a mapping $\phi:V(G)\to V(G')$ such that
$(G,S,\Lambda,k)^\phi=(G',S',\Lambda',k')$
where we apply $\phi$
as defined in the preliminaries).
\begin{proof}
We consider two cases depending on the size of $S$.
\begin{cs}
\case{$|S|\leq \funcsmall$}
Let $\CI\coloneqq\{(\{x\},\{y\})~|~x,y\in S,x\neq y,xy\nin E(G)\}$
and let
\[ D=S\cup\bigcup_{(L,R)\in\CI}S_{L,R}(w_{L\cup R,k+1}).\]
We set $w\coloneqq w_{S,k+1}$ and then we have the
following for every vertex set $Z$ of a connected component of $G-D$ by \cref{lem:union:of:weighted:minimal:separations}:
\begin{align*}
|N(Z)|&\leq w(N(Z))
\overset{\mathclap{\ref{lem:union:of:weighted:minimal:separations}}}{\leq} w(S)
\leq \funcsmall(k+1)
\leq \funcmedium.
\end{align*}
For every $xy \nin E(G)$ there is
a~$(\{x\},\{y\})$-separator of size at most~$k$ disjoint
from~$\{x,y\}$, because $G$ is $k$-improved.
Thus $|D|\le
|S|+k|S|^2\le\funcsmall+k\funcsmall^2\le\funclarge$.
Moreover,
since $G-S$ is
connected and $S = N_G(V (G) \setminus S)$, for all distinct $x,y\in S$ every~minimum weight~$(\{x\},\{y\})$-separator contains a vertex that is not in~$S$. It follows that~$D\neq S$.

\case{$\funcsmall<|S|\leq \funcmedium$}
Let $\CI\coloneqq\{(L,R)~|~|L|=|R|\leq
k+2,|S_{L,R}|\leq k+1\}$ and let
\[
D=S\cup\bigcup_{(L,R)\in \CI}S_{L,R}.\]
\end{cs}
The properties of~$D$ follow from similar arguments as in the first case.
The fact that $\CI$ is nonempty
follows from the existence of a \emph{balanced} separation
(for details see \cite{lokshtanov2017fixed}).
Next, we show how to find $\hat\Lambda$ in both cases.
To each $x\in D\setminus S$ we associate
the set $A_x\coloneqq\{(L,R)\in\CI\mid x\in S_{L,R}\}$.
Two vertices $x$ and $y$ occur in exactly the same separators if
$A_x=A_y$. In this case we call them equivalent and write $x\equiv y$.
Let $A_1,\ldots,A_t\subseteq D\setminus S$ be the equivalence classes of
``$\equiv$''. Since
each $x$ is contained in some separator of size
at most $k+1$ we conclude that the size of each $A_i$
is at most $k+1$.

For each labeling $\lambda\in\lab(S)$ we choose an extension
$\hat\lambda:D\to \{1,\ldots,|D|\}$ such that $\hat\lambda|_S=\lambda$
and for~$x,y\in D\setminus S$ we have $x^{\hat\lambda} <y^{\hat\lambda}$ if $A_x^\lambda\prec A_y^\lambda$. (Recall that~$\prec$ is the linear order of subsets of~$\mathbb{N}$ as defined in the preliminaries).
Inside each equivalence class $A_i$, the ordering
is chosen arbitrarily.
Define $\hat\Lambda\coloneqq(\{\id_{S}\}\times \sym(A_1)\times\ldots\times\sym(A_t)) \cdot  \{\hat\lambda\mid
\lambda\in\Lambda\}\leq\lab(D)$.
By construction the coset~$\hat\Lambda$ does not depend on the choices of the
extensions~$\hat\lambda$. Since $|A_i|\leq k+1$ for all $1\leq i\leq t$ we
conclude that $\hat\Lambda\in\gammak$, as desired.

It remains to explain how to efficiently compute~$\hat\Lambda$. For
this we simply remark that is suffices to use a set of
extensions~$M\subseteq \Lambda$ such that~$\Lambda$ is the smallest
coset containing all elements of~$M$ (i.e., we can use a coset
analogue of a generating set).
We conclude that~$\hat\Lambda$ can be computed in polynomial time in the size of~$\CI$.

Last but not least, 
we show how to construct the graph $H$.
The Case $|S|\leq \funcsmall$ is easy to handle.
In this case we define $H$ as the complete graph
on the set $D\cup\{r\}$
where $r$ is some new vertex, which becomes the root
of $H$.
The forward-degree of $r$ is bounded by
$|D|$ which in turn
is bounded by $k^{\CO(1)}$.
We consider the Case $\funcsmall<|S|\leq \funcmedium$.
We define
$V(H):=\{(L,R)\mid L,R\subseteq S, |L|=|R| \leq k+2\}\cup D$.
Clearly, we have $\CI\subseteq V(H)$.
For the root we choose $(\emptyset,\emptyset)\in V(H)$.
We define the edges $E(H):=\{(L,R)(L',R')\mid L\subseteq L',
R\subseteq R', |L|+1=|R|+1=|L'|=|R'|\}
\cup \{x(L,R)\mid x\in D,x\in S_{L,R}\}$.
Since for each pair $(L,R)$
there are at most $|S|^2$ different
extensions $(L',R')$ with $|L|+1=|R|+1=|L'|=|R'|$
and
since each separator $S_{L,R}$ contains
at most $k+1$ vertices, we conclude that
the forward-degree of each vertex in $H$ is
bounded by $|S|^2+k+1\in k^{\CO(1)}$.
Moreover, the depth of $H$ is bounded by $k+3$.
\end{proof}

A \emph{labeled tree decomposition} $(T,\beta,\alpha,\eta)$
is a 4-tuple where $(T,\beta)$ is a tree decomposition and
$\alpha$ is a function that maps each $t\in V(T)$ to a labeling
coset $\alpha(t)\leq\lab(\beta(t))$
and $\eta$ is a function that maps each $t\in V(T)$
to a graph $\eta(t)$.

The lemma can be used as a recursive tool to compute our desired
isomorphism-invariant labeled tree decomposition.

\begin{theo}\label{theo:labelDecomposition}
Let $k\in\NN$ and let $G$ be a $k$-basic graph
and let
$v$ be a vertex of degree at most $k$.
There is an algorithm that
either correctly concludes that $\tw(G) > k$,
or computes a labeled
tree decomposition $(T,\beta,\alpha,\eta)$
with the following properties.
\begin{enumerate}[label=(\Roman*)]
\item\label{i:2} the width of $(T,\beta)$ is bounded by $\funclarge$,
\item\label{i:3} the adhesion width of $(T,\beta)$ is bounded by $\funcmedium$,
\item\label{i:4} the degree of $(T,\beta)$ is bounded by $k\funclarge^2$
and the number of children of $t$ with common adhesion set is bounded
by $k$ for each $t\in V(T)$,
\item\label{i:5} $|V(T)|$ is bounded by $\CO(|V(G)|)$,
\item\label{i:7} for each bag~$\beta(t)$ the adhesion sets of the children are all equal
to~$\beta(t)$ or the adhesion sets of the children are all distinct.
In the former case the bag size is bounded by $\funcmedium$,
\item\label{i:8}
for each $t \in V(T)$ the graph $\eta(t)=H_t$ is a connected rooted graph
such that
$\beta(t)\cup\beta(t)^2\subseteq V(H_t)$
and for each adhesion set $S$
there is a corresponding vertex
$S\in V(H_t)$,
$\depth(H_t)\in\CO(k)$
and
$\fdeg(v)\in k^{\CO(1)}$ for
all $v\in V(H_t)$, and
\item\label{i:6} $\alpha(t)\in\gammak$.
\end{enumerate}
The algorithm runs in time $2^{\CO(k^2 \log k)}|V (G)|^{\CO(1)}$ and the
output~$(T,\beta,\alpha,\eta)$ of the algorithm is isomorphism invariant
(w.r.t.~$G,v$ and~$k$).
Furthermore,
if we drop Property \ref{i:6} as a requirement, the triple  $(T,\beta,\eta)$
can be computed in
time $2^{\CO(k \log k)}|V (G)|^{\CO(1)}$. 
\end{theo}

Here, the output of an algorithm $\CA$
is isomorphism-invariant,
if all isomorphisms between two input data
extend to an isomorphism between the output.
More precisely,
an isomorphism $\phi\in\iso((G,v,k),(G',v',k'))$
extends to an isomorphism between
between $(T,\beta,\alpha,\eta)$ and $(T',\beta',\alpha',\eta')$
if there is a bijection between the tree decompositions
$\phi_T:V(T)\to V(T')$ and for each
node $t\in V(T)$ a bijection
between the vertices of graphs $\phi_t:V(\eta(t))\to V(\eta'(\phi_T(t)))$
which extend $\phi$, i.e.\
$\phi_t(x)=x^\phi$ for all $x\in \beta(t)\cup\beta(t)^2\cup 2^{\beta(t)}$
where we naturally apply $\phi$ as defined in the preliminaries.
Furthermore, these extensions define an isomorphism between the output data,
i.e.\ for all nodes $t\in V(T)$ we have that
$\beta(t)^\phi=\beta'(\phi_T(t))$,
$\alpha(t)^{\phi} = \alpha'(\phi_T(t))$
and
$\eta(t)^{\phi_t}=\eta'(\phi_T(t))$.

\begin{proof}
We describe a recursive algorithm $\CA$ that has
$(G,S\subseteq V(G),\Lambda\leq\lab(S))$
as the input such that the data satisfy the assumptions
of~Lemma~\ref{lem:B}. That is:\\
\begin{minipage}[t]{0.45\textwidth}
\begin{enumerate}
\item\label{e:1} $\emptyset\subsetneq S \subsetneq V (G)$,
\item\label{e:2} $|S|\leq \funcmedium$,
\item\label{e:3} $S$ is not a clique or $|N(S)|\leq k$,
\end{enumerate}
\end{minipage}
\begin{minipage}[t]{0.5\textwidth}
\begin{enumerate}
\setcounter{enumi}{3}
\item\label{e:4} $G- S$ is connected,
\item\label{e:5} $S = N_G(V (G) \setminus S)$ and,
\item\label{e:6} $\Lambda\in\gammak$.
\end{enumerate}
\end{minipage}\\[1.5ex]
The output is a labeled
tree decomposition $(T,\beta,\alpha,\eta)$ with root bag $S$.
Since $G$ has no separating cliques and therefore $v$ is not a cut vertex,
the triple $(G,\{v\},\{v\mapsto 1\})$ fulfills the input conditions
which will complete the proof.
The algorithm $\CA$ works as follows.
\begin{stp}
\step1
We describe how to construct $\beta(t)=D$ and $\alpha(t)=\hat\Lambda$
for the root node $t\in V(T)$.
If $S$ is a clique we define
$D\coloneqq S\cup N(S)$.
Notice that $N(S)\leq k$ due to Property \ref{e:3}.
Let $\tau:N(S)\to\{|S|+1,\ldots,|D|\}$
be an arbitrary bijection.
We define
$\hat\Lambda\coloneqq\Lambda\times\tau\sym(\{|S|+1,\ldots,|D|\})\leq\lab(D)$.
Otherwise (i.e., $S$ is not a clique) we compute and define $D$ and
$\hat\Lambda$ as in \cref{lem:B}.

\end{stp}
Let $Z_1,\ldots,Z_r$ be the connected components of 
the graph $G- D$.
Let $S_i\coloneqq N_G(Z_i)$ and $G_i\coloneqq G[Z_i\cup S_i]$.

\begin{stp}

\step2 We describe how to construct the graph $\eta(t)=H_t$.
First, we compute the graph $H$ as in \cref{lem:B}.
To achieve that pairs of vertices from $\beta(t)$
are contained in this graph we define a Cartesian product $H^2$,
i.e. $V(H^2)=V(H)^2$ and
\begin{align*}
 E(H^2) =   \;\;\;\;\;&\{(u_1,u_2)(v_1,v_2)\mid u_1v_1\in E(H) \text{ and }u_2=v_2\}\\
                 \cup &\{(u_1,u_2)(v_1,v_2)\mid u_2v_2\in E(H) \text{ and }u_1=v_1\}.
\end{align*}
Now, we define $H_t$ as follows.
The vertex set is $V(H_t):=V(H)\cup V(H^2)\cup \{S_i\mid i\in[r]\}$.
The edge set is $E(H_t):=E(H^2)\cup \{(v,v)v\mid v\in H\}
\cup\{(u,v)S_i\mid u,v\in S_i, uv\nin E(D)\}$
and the root of $H_t$ is $(r,r)$ where $r$ is the root of $H$.

\step3
We call the algorithm recursively as follows.
Let \[\Lambda_i\coloneqq\{\lambda\in\hat\Lambda|_{S_i}~|~
S_i^\lambda=\min_{\prec}\{S_i^\lambda~|~\lambda\in\hat\Lambda\}\}.\]
(This set is essentially only an isomorphism-invariant restriction of~$\Lambda$ to~$S_i$.
The minimum is taken with respect to~$\prec$, the ordering of subsets of~$\mathbb{N}$ as defined in the preliminaries.) 
The image~$S_i^\lambda$ might be a set of natural numbers different from~$\{1,\ldots,|S_i|\}$.
To rectify this we let~$\pi_i$ be the bijection from~$S_i^{\lambda}$ to~$\{1,\ldots,|S_i|\}$ that preserves the normal ordering~``$<$'' of natural numbers and set~$\Lambda'_i = \Lambda_i \pi_i$.
We compute $\CA(G_i,S_i,\Lambda_i')=(T_i,\beta_i,\alpha_i,\eta_i)$ recursively.
By possibly renaming the vertices of the trees~$T_i$, we can assume w.l.o.g. $V(T_i)\cap V(T_j)=\emptyset$.

\step4 We define a labeled tree decomposition $(T,\beta,\alpha,\eta)$
as follows. Define a new root vertex $t$ and attach the root vertices
$t_1,\ldots,t_r$ computed by the recursive calls as children.
We define $\beta(t)\coloneqq D$ and $\alpha(t)\coloneqq\hat\Lambda$
and $\eta(t):=H_t$.
Combine all decompositions $(T_i,\beta_i,\alpha_i,\eta_i)$ together with the
new root $t$ and return the resulting decomposition
$(T,\beta,\alpha,\eta)$.

\end{stp}

\begin{correctness}
First of all, we show that $(G_i,S_i,\Lambda_i)$
fulfills the Conditions \ref{e:1}-\ref{e:6} and the recursive call in Step 3
is justified.
Condition \ref{e:1} is clear.
For \ref{e:2} notice that $D$ and $\hat\Lambda$ satisfy
\ref{a:a}-\ref{a:d} from \cref{lem:B}
regardless which option is applied in Step 1.
Condition \ref{e:2} then follows from \ref{a:c}.
For \ref{e:3} notice that $G$ does not have any separating cliques,
but $S_i$ is a separator. The Conditions \ref{e:4} and \ref{e:5} follow 
from the definition of $G_i$ and $S_i$, whereas
\ref{e:6} follows from~\ref{a:d}.

Next, we show \ref{i:2}-\ref{i:7}.
The bound for the width in \ref{i:2} immediately follows from
\ref{a:b}, while the bound on the adhesion width in \ref{i:3}
follows from~\ref{a:c}.
For \ref{i:4} observe that
$G$ is $k$-improved and therefore
each non-edge in $D$ is contained in at most $k$
different $S_i$, i.e., $|\{i~|~u,v\in S_i\}|\leq k$ for each
$u,v\in D,uv\nin E(D)$. Therefore, $r\leq k|\bar E(D)|\leq k\funclarge^2$.
Moreover, since each $S_j$ contains a non-edge $uv$ we conclude
$|\{i\mid S_i=S_j\}|\leq |\{i~|~u,v\in S_i\}|\leq k$.
For \ref{i:5} we associate each bag $\beta(t)=D$ with the
set $D\setminus S$ in an injective way.
The number of sets $D\setminus S$ occurring in each recursive call
is bounded by $|V(G)|$ since these sets
are not empty and pairwise distinct. We conclude $|V(T)|\leq |V(G)|$.
Property \ref{i:6} follows from \ref{a:d}.
For any bag not satisfying Property \ref{i:7}
we introduce
a new bag containing all equal adhesion sets.
For \ref{i:8} we observe that the forward-degree
of the vertices $(u,v)\in V(H^2)$ is
the sum of the forward-degrees of $u$ and $v$ in $H$.
The depth of $H^2$ is twice the depth of $H$.
Therefore, we have $\fdeg(x)\in k^{\CO(1)}$ for
each $x\in V(H_t)$ and $\depth(H_t)\leq 2k+7$.
The construction is obviously isomorphism invariant.
\end{correctness}

\begin{runningtime}
We have already seen that the number of recursive calls is bounded by $|V(G)|$.
The computation of the sets $D$ and $\hat\Lambda$ as in \cref{lem:B}
can be done in time $2^{\CO(k \log k)}|V (G)|^{\CO(1)}$.
We need to explain how to compute $\Lambda_i$ from $\hat\Lambda$.
With standard group theoretic techniques, including 
the Schreier-Sims algorithm, one can find a minimal image
of $S_i$.
Using set transporters one can find all labelings in $\hat\Lambda$
mapping to this minimal image. This algorithm runs
in time $|D|^{\CO(k)}\subseteq 2^{\CO(k^2\log k)}$.
Since we are only interested in an isomorphism invariant
restriction of $\hat\Lambda$ to $S_i$, we could
also use techniques from Section~\ref{sec:canoniztion:tools} as follows.
We define a hypergraph $X$ on the vertex set $D$ consisting of one single
hyperedge $S_i$.
Let $\Lambda_i\coloneqq\can(X;\hat\Lambda)|_{S_i}$ be the restriction
of a canonical
labeling of that hypergraph.
This algorithm also runs in time $|D|^{\CO(k)}$.
This leads to a total run time of $2^{\CO(k^2 \log k)}|V (G)|^{\CO(1)}$.
We remark that the computation of $\Lambda_i$ from $\hat\Lambda$
is only needed if we require Property \ref{i:6}.
For the computation of a labeled tree decomposition $(T,\beta,\eta)$
the bottleneck arises in \cref{lem:B} which gives a run time of
$2^{\CO(k \log k)}|V (G)|^{\CO(1)}$.
\end{runningtime}
\end{proof}

\begin{rem}\label{rem:sub}
We later use the isomorphism-invariance of the
labeled tree decomposition
$(T,\beta,\alpha,\eta)$ from the previous theorem in more detail.
Let $t\in V(T)$ be a non-root node
and let $S\subseteq V(G)$ be the adhesion set to the parent node of $t$
and let $I_t=(T_t,\beta_t,\alpha_t,\eta_t)$ be the decomposition of
the subtree rooted at $t$ and $G_t$ the graph corresponding to $I_t$.
Then $\eta_t$ is
isomorphism-invariant w.r.t. $T_t,\beta_t,G_t$ and $S$.
\end{rem}


\section{Coset-Hypergraph-Isomorphism}\label{sec:iso:sub}

After having computed isomorphism-invariant tree decompositions in the previous sections we now want to compute the set of isomorphisms from one graph to another in a bottom up fashion.
Let $G_1,G_2$ be the two input graphs and suppose we are given isomorphism-invariant tree decompositions $(T_1,\beta_1)$ and $(T_2,\beta_2)$.
For a node $t \in V(T_i)$ we let $(G_i)_t$ be graph induced by the union of all bags contained in the subtree rooted at $t$.
The basic idea is to compute for all pairs $t \in V(T_1),t'\in V(T_2)$ the set of isomorphisms from $(G_1)_t$ to $(G_1)_{t'}$ (in addition the isomorphisms shall also respect the underlying tree decomposition) in a bottom up fashion.

The purpose of this section is to give an algorithm that solves this problem at a given bag (assuming we have already solved the problem for all paris of children of $t$ and $t'$).
Let us first give some intuition for this task.
Suppose we are looking for all bijections from $\beta_1(t)$ to $\beta_2(t')$ that can be extended to an isomorphism from $(G_1)_t$ to $(G_2)_{t'}$.
Let $t_1,\dots,t_{\ell}$ be the children of $t$ and $t_1',\dots,t_{\ell}'$ the children of $t'$.
Then we essentially have to solve the following two problems.
First, we have to respect the edges appearing in the bags $\beta_1(t)$ and $\beta_2(t')$.
But also, every adhesion set $\beta(t) \cap \beta(t_i)$ has to be mapped to another adhesion set $\beta(t') \cap \beta(t_j')$
in such a way that the corresponding bijection (between the adhesion sets) extends to an isomorphism from $(G_1)_{t_i}$ to $(G_2)_{t_j'}$.
In order to solve this problem we
first consider the case in which the adhesion sets are all distinct
and define the following abstraction.

An instance of \emph{coset-hypergraph-isomorphism} is a tuple $\mathcal{I} = (V_1,V_2,\mathcal{S}_1,\mathcal{S}_2,\chi_1,\chi_2,\CF, \Cf)$ such that
\begin{enumerate}
 \item $\mathcal{S}_i \subseteq \Pow(V_i)$,
 \item $\chi_i\colon \mathcal{S}_i \rightarrow \mathbb{N}$ is a coloring,
 \item $\CF = \{\Theta_S \leq \sym(S) \mid S \in \mathcal{S}_1\}$, and
 \item $\Cf = \{\tau_{S_1,S_2}\colon S_1 \rightarrow S_2 \mid S_1 \in
 \mathcal{S}_1,S_2 \in \mathcal{S}_2 \text{ such that } \chi_1(S_1) =
 \chi_2(S_2)\}$ such that
  \begin{enumerate}
   \item every $\tau_{S_1,S_2} \in \Cf$ is bijective, and
   \item for every color $i \in \mathbb{N}$ and every $S_1,S_1' \in \chi_1^{-1}(i)$ and $S_2,S_2' \in \chi_2^{-1}(i)$ and $\theta \in \Theta_{S_1},\theta' \in \Theta_{S_1'}$ it holds that 
    \begin{equation}
     \label{eq:n-property}\tag{N}
     \theta\tau_{S_1,S_2}(\tau_{S_1',S_2})^{-1}\theta'\tau_{S_1',S_2'} \in \Theta_{S_1}\tau_{S_1',S_2}.
    \end{equation}
  \end{enumerate}
\end{enumerate}
The instance is \emph{solvable} if there is a bijective mapping $\phi\colon V_1 \rightarrow V_2$ such that
\begin{enumerate}
 \item $S \in \mathcal{S}_1$ if and only if $S^{\phi} \in \mathcal{S}_2$ for all $S \in \Pow(V_1)$,
 \item $\chi_1(S) = \chi_2(S^{\phi})$ for all $S \in \mathcal{S}_1$, and
 \item for every $S \in \mathcal{S}_1$ it holds that $\phi|_S \in \Theta_S\tau_{S,S^{\phi}}$.
\end{enumerate}
In this case we call $\phi$ an \emph{isomorphism} of the instance $\mathcal{I}$.
Moreover, let $\iso(\mathcal{I})$ be the set of all isomorphisms of
$\mathcal{I}$.
Observe that property (\ref{eq:n-property}) describes a consisteny condition: if we can use $\sigma_1$ to map $S_1$ to $S_2$, $\sigma_2$ to map $S_1'$ to $S_2$, and $\sigma_3$ to map $S_1'$ to $S_2'$, then mapping $\sigma_1\sigma_2^{-1}\sigma_2$ can be used to map $S_1$ to $S_2'$.
As a result the set of all isomorphisms of the instance $\mathcal{I}$ forms a coset,
that is $\iso(\mathcal{I}) = \Theta \phi$ for some $\Theta \leq \sym(V_1)$ and $\phi \in \iso(\mathcal{I})$.

In the application in the main recursive algorithm, the sets $V_i = \beta(t_i)$, the hyperedges $\mathcal{S}_i$ are the adhesion sets of $t_i$ (and we will also encode the edges appearing in the bag in this way),
and the cosets $\Theta_{S_1}\tau_{S_1,S_2}$ tell us which mappings between the adhesion sets $S_1$ and $S_2$ extend to an isomorphism between the corresponding subgraphs.
The colorings $\chi_1$ are used to indicate which subgraphs can not be mapped to each other (and also to distinguish between the adhesion sets and egdes of the bags which will both appear in the set of hyperedges).

The next Lemma gives us one of the central subroutines for our recursive algorithm.

\begin{lem}
 \label{lem:big_iso}
 Let $\mathcal{I} = (V_1,V_2,\mathcal{S}_1,\mathcal{S}_2,\chi_1,\chi_2,\CF, \Cf)$ be an instance of coset-hyper\-graph-iso\-mor\-phism.
 Moreover, suppose there are isomorphism-invariant rooted graphs $H_1 = (W_1,E_1,r_1)$ and $H_2 = (W_2,E_2,r_2)$ such that
 \begin{enumerate}
  \item $V_i \cup \mathcal{S}_i \subseteq W_i$,
  \item $\fdeg(w) \leq d$ for all $w \in W_i$,
  \item $\depth(H_i) \leq h$, and
  \item $|S| \leq d$ for all $S \in \mathcal{S}_i$
 \end{enumerate}
 for all $i \in \{1,2\}$.
 Then a representation for the set $\iso(\mathcal{I})$ can be computed in time
 $2^{\CO(h \cdot \log^{c}d)}$ for some constant $c$.
\end{lem}

Here, isomorphism-invariant means that for every isomorphism $\phi \in \iso(\mathcal{I})$ there is an isomorphism $\phi_H$ from $H_1$ to $H_2$
such that $v^{\phi} = v^{\phi_H}$ for all $v \in V_1$ and $\{v^{\phi} \mid v \in S\} = S^{\phi_H}$ for all $S \in \mathcal{S}_1$.

The proof of this lemma is based on the following theorem.
A rooted tree $T = (V,E,r)$ is \emph{$d$-ary} if every node has at most $d$ children.
An expanded rooted tree is a tuple $(T,C)$ where $T=(V,E,r)$ is a rooted tree and $C \colon L(T)^{2}\rightarrow \rg(C)$ is a coloring of pairs of leaves of $T$ ($L(T)$ denotes the set of leaves of $T$).
Isomorphisms between expanded trees $(T,C)$ and $(T,C')$ are required to respect the colorings $C$ and $C'$.

\begin{theo}[\cite{bounded-degree-small-diameter}]\label{theo:neuen}
 Let $(T,C)$ and $(T',C')$ be two expanded $d$-ary trees and let $\Gamma \leq \aut(T)$ and $\gamma \in \iso(T,T')$.
 Then one can compute a representation of the set $\{\phi \in \Gamma\gamma \mid (T,C)^{\phi} = (T',C')\}$ in time $n^{\mathcal{O}(\log^{c}d)}$ for some constant $c$.
\end{theo}

\begin{proof}[Proof of Lemma \ref{lem:big_iso}]
 For $i \in \{1,2\}$ let $H_i'$ be the graph obtained from $H_i$ by adding, for each $S_i \in \mathcal{S}_i$, vertices $(S_i,v)$ for every $v \in S_i$ connected to the vertex $S_i$.
 Note that $\depth(H_i') \leq h+1$ and $\fdeg(w) \leq 2d$ for all $w \in V(H_i')$.
 Let $r_i$ be the root of $H_i'$.
 
 A \emph{branch} of the graph $H_i'$ is a sequence of vertices $(v_0,\dots,v_k)$
 such that $v_{j-1}v_{j} \in E(H_i')$ for every $j \in [k]$ and $\dist(r_i,v_j) = j$ for all $j \in \{0,\dots,k\}$ (in particular $v_0 = r_i$).
 We construct a rooted $2d$-ary tree $T_i$ as follows.
 Let $V(T_i)$ be the set of all branches of $H_i'$ and define $E(T_i)$ to be the set of all pairs $\bar v\,\bar v'$ such that $\bar v = (v_0,\dots,v_k)$ and $\bar v' = (v_0,\dots,v_k,v_{k+1})$.
 The root of $T_i$ is $(v_0)$ where $v_0$ is the root of $H_i'$.
 It is easy to check that $|V(T_i)| \leq (2d+1)^{h+1}$.

 For a sequence of vertices $\bar v = (v_1,\dots,v_k)$ let $\last(\bar v) = v_k$ be the last entry of the tuple $\bar v$.
 Let $\chi_i^{'} \colon V(T_i) \rightarrow \mathbb{N}$ be the coloring defined by $\chi_i'(\bar v) = 1$ for all $\bar v$ where $\last(\bar v) \in V_i$,
 $\chi_i'(\bar v) = 2$ for all $\bar v$ where $\last(\bar v) \in V(H_i') \setminus V(H_i)$,
 $\chi_i'(\bar v) = 3 + \chi_i(\last(\bar v))$ for all $\bar v$ where $\last(\bar v) \in \mathcal{S}_i$, and $\chi'_i(\bar v) = 0$ for all other tuples.
 Let $\Gamma_i := \aut(T_i,\chi_i')$.
 Also let $\gamma_{1 \rightarrow 2}$ be an isomorphism from $(T_1,\chi_1')$ to $(T_2,\chi_2')$ (if no such isomorphism exists then $\iso(\mathcal{I}) = \emptyset$ and we are done).
 
 Let $A_i = \{\bar v \in V(T_i) \mid \last(\bar v) \in V(H_i') \setminus V(H_i)\}$.
 Note that every element $\bar v \in A_i$ is a leaf of $T_i$.
 Also note that $A_i$ is invariant under $\Gamma_i$, that is, $A_i^{\gamma} = A_i$ for all $\gamma \in \Gamma_i$.
 For $\bar v \in A_i$ let $\parent(\bar v)$ be the parent node of $\bar v$ in the tree $T_i$.
 Observe that $\last(\parent(\bar v)) \in \mathcal{S}_i$ for every $\bar v \in A_i$.
 Let $B_i = \{\parent(\bar v) \mid \bar v \in A_i\}$.
 For $\bar w \in B_i$ let $A_i(\bar w) = \{\bar v \in A_i \mid \parent(\bar v) = \bar w\}$
 and let $A_i^{*}(\bar w) = \{\last(\bar v) \mid \bar v \in A_i(\bar w)\} = \{S\} \times S$ where $S = \last(\bar w)$.
 
 Now let $\gamma \in \Gamma_1\gamma_{1 \rightarrow 2}$ be an isomorphism from $(T_1,\chi_1')$ to $(T_2,\chi_2')$.
 Then, for every $\bar w \in B_1$ the isomorphism $\gamma$ induces a bijection $\gamma^{\bar w}\colon \last(\bar w) \rightarrow \last(\bar w^{\gamma})$ which is induced by how $\gamma$ maps $A_1(\bar w)$ to $A_2(\bar w^{\gamma})$.
 Now let
 \[\Gamma_1'\gamma_{1 \rightarrow 2}' = \{\gamma \in \Gamma_1\gamma_{1 \rightarrow 2} \mid \forall \bar w \in B_1\colon \gamma^{\bar w} \in \Theta_{\last(\bar w)}\tau_{\last(\bar w),\last(\bar w^{\gamma})}\}.\]
 Observe that the set indeed forms a coset due to property (\ref{eq:n-property}).
 Moreover, a representation for $\Gamma_1'\gamma_{1 \rightarrow 2}'$ can be computed in time polynomial in $|V(T_1)|$ since the group $\Gamma_1$ acts independently as a symmetric group on all the sets $A_1(\bar w)$.
 
 Observe that, up to now, we have not enforced any consistency between the maps
 $\gamma^{\bar w}$.
 Indeed, there may be elements $\bar w, \bar w' \in B_1$ such that $\last(\bar w) = \last(\bar w')$ but $\gamma^{\bar w} \neq \gamma^{\bar w'}$.
 To finish the proof we need to enforce such a consistency property.
 Then every isomorphism $\gamma \in \Gamma_1'\gamma_{1 \rightarrow 2}'$ that fulfills the consistency property naturally translates back to an isomorphism of $\mathcal{I}$.
 
 For $v \in V(H_i')$ we define the label $\ell(v) = v$ if $v \in V(H_i)$ and $\ell(v) = w$ if $v = (S,w) \in V(H_i') \setminus V(H_i)$.
 Also we define colorings $C_i\colon L(T_i)^{2} \rightarrow \mathbb{N}$ in such a way that for all
 $\bar v=(v_0,\ldots,v_k) ,\bar v'=(v_0',\ldots,v_{k'}') \in L(T_i)$,
 $\bar w=(w_0,\ldots,w_\ell), \bar w'=(w_0',\ldots,w_{\ell'}')\in L(T_j)$
 \begin{align}
  \label{eq:1}
  C_i(\bar v,\bar v')=C_j(\bar w,\bar w')\iff &k=\ell\text{ and } k'=\ell'\\
    \label{eq:2}
    &\text{and }\forall i\in[k],j\in[k']:\big(\ell(v_i)=\ell(v_j')\iff \ell(w_i)=\ell(w_j')\big)\\
    \label{eq:3}
    &\text{and }\bar v \in A_i \iff \bar w \in A_j \text{ and } \bar v' \in A_i \iff \bar w' \in A_j.
 \end{align}
 Now let
 \[\Gamma_1^{*}\gamma_{1 \rightarrow 2}^{*} = \{\gamma \in \Gamma_1'\gamma_{1 \rightarrow 2}' \mid (T_1,C_1)^{\gamma} = (T_2,C_2)\}.\]
 By Theorem \ref{theo:neuen} a representation for this coset can be computed in time $2^{\mathcal{O}(h \cdot \log^{c}d)}$ for some comstant $c$.
 We claim that every $\gamma \in \Gamma_1^{*}\gamma_{1 \rightarrow 2}^{*}$ naturally defines an isomorphism of $\mathcal{I}$ and conversely, every isomorphism of $\mathcal{I}$ can be described like this.
 Let $V_i' = \{\bar v \in V(T_i) \mid \last(\bar v) \in V_i\}$.
 Let $\bar v,\bar v' \in V_i'$ such that $\last(\bar v) = \last(\bar v')$ and let $\gamma \in \Gamma_1^{*}\gamma_{1 \rightarrow 2}^{*}$.
 Then $\last(\bar v^{\gamma}) = \last((\bar v')^{\gamma})$ by property (\ref{eq:2}) and (\ref{eq:3}) of the colorings $C_i$.
 Hence, every $\gamma \in \Gamma_1^{*}\gamma_{1 \rightarrow 2}^{*}$ naturally defines a bijection $\gamma^{V_1}\colon V_1 \rightarrow V_2$.
 
 Additionally, for $\bar v \in A_i$ and $\bar v' \in V_i'$ such that
 $\ell(\last(\bar v)) = \last(\bar v')$ and $\gamma \in \Gamma_1^{*}\gamma_{1 \rightarrow 2}^{*}$ it holds that $\last(\bar v^{\gamma}) = \last((\bar v')^{\gamma})$ by the same argument.
 With this, we obtain that $\gamma^{V_1}$ is indeed an isomorphism of $\mathcal{I}$.
 The hyperedges $\mathcal{S}_1$ and $\mathcal{S}_2$ are encoded in the tree structure (and they have to be mapped to each other by the coloring $\chi_i'$)
 and the coloring of hyperedges is encoded in the coloring $\chi_i'$ which has to be preserved.
 Also, when mapping an element $\bar w \in B_1$ to another element $\bar w^{\gamma} \in B_2$ the children have to mapped to each other in such a way that the corresponding bijection lies in the corresponding coset.
 So overall $\gamma^{V_1}$ is an isomorphism of $\mathcal{I}$.
 
 On the other hand, every isomorphism $\phi \in \iso(\mathcal{I})$ extends naturally to an isomorphism from $H_1$ to $H_2$ and thus,
 it also extends to an isomorphism from $H_1'$ to $H_2'$ and from $T_1$ to $T_2$, since all relevant objects are defined in an isomorphism-invariant way.
 Hence, from a representation for $\Gamma_1^{*}\gamma_{1 \rightarrow 2}^{*}$ we can easily compute a representation for $\iso(\mathcal{I})$.
\end{proof}

Looking at the properties of the tree decompositions computed in Theorem \ref{theo:cliqueDecomposition} and \ref{theo:labelDecomposition} we have for every node $t$ that either the adhesion sets to the children are all equal or they are all distinct.
Up to this point we have only considered the problem that all adhesion sets are distinct (i.e.\ the coset-hypergraph-isomorphism problem).
Next we consider the case that all adhesion sets are equal.
Towards this end we define the following variant.

An instance of \emph{multiple-colored-coset-isomorphism}
is a $6$-tuple
$\mathcal{I} = (V_1,V_2,\chi_1,\chi_2,\CF, \Cf)$ such that
\begin{enumerate}
 \item $\chi_i\colon [t] \rightarrow \mathbb{N}$ is a coloring,
 \item $\CF = \{\Theta_i \leq \sym(V) \mid i\in[t]\}$, and
 \item $\Cf = \{\tau_{i,j}\colon V_1 \rightarrow V_2 \mid i,j\in[t]
 \text{ such that } \chi_1(i) = \chi_2(j)\}$ such that
  \begin{enumerate}
   \item every $\tau_{i,j} \in \Cf$ is bijective, and
   \item for every color $i \in \mathbb{N}$ and every $j_1,j_1'\in
   \chi_1^{-1}(i)$ and $j_2,j_2' \in \chi_2^{-1}(i)$ and $\theta \in
   \Theta_{j_1},\theta' \in \Theta_{j_1'}$ it holds that
    \begin{equation}
     \label{eq:n-property-small}\tag{N2}
     \theta\tau_{j_1,j_2}\tau_{j_1',j_2}^{-1}\theta'\tau_{j_1',j_2'} \in
     \Theta_{j_1}\tau_{j_1',j_2}.
    \end{equation}
  \end{enumerate}
\end{enumerate}
The instance is \emph{solvable} if there is a bijective mapping $\phi\colon V_1 \rightarrow V_2$
and a $\pi\in\sym(t)$ such that
\begin{enumerate}
 \item $\chi_1(i) = \chi_2(\pi(i))$ for all $i \in [t]$, and
 \item for every $i \in [t]$ it holds that $\phi\in
 \Theta_i\tau_{i,\pi(i)}$.
\end{enumerate}
The set $\iso(\CI)$ is defined analogously.

\begin{lem}
 \label{lem:small_iso}
 Let $\mathcal{I} = (V_1,V_2,\chi_1,\chi_2,\CF, \Cf)$
 be an instance of mul\-ti\-ple-col\-ored-coset-iso\-mor\-phism.
 Then a representation for the set $\iso(\mathcal{I})$ can be computed in
 time
 \[\min((|V_1|!|\CF|)^{\CO(1)},|\CF|!^{\CO(1)}2^{\CO(\log^c(|V_1|))})\]
 for some constant $c$.
\end{lem}

\begin{proof}
The first run time can be achieved by brute force in the following way.
We iterate through all $\phi\colon V_1 \rightarrow V_2$
and check in polynomial time
if a $\pi\in\sym(t)$ exists that satisfies (1) and (2).
(This can be done by computing, for every $i \in [t]$, the set $A_i$ of those $j \in [t]$ such that $\phi \in \Theta_i\tau_{i,j}$.
Then one needs to find a permutation $\pi \in \sym(t)$ such that $\pi(i) \in A_i$ for all $i \in [t]$. This can be interpreted as a matching problem.)
The results is the union of all these $\phi$.

However, for the second run time we
iterate through all $\pi\in\sym(t)$ satisfying (1) and compute
all corresponding $\phi\colon V_1 \rightarrow V_2$
satisfying (2)
which in turn can be done by
iterated coset-intersection.
The result is
$\iso(\CI)=\bigcup_{\pi\in\sym(t)\text{ satisfying (1)}}
\bigcap_{i\in[t]}\Theta_i\tau_{i,\pi(i)}$.
Coset-intersection can be done in quasi-polynomial time
\cite{DBLP:conf/stoc/Babai16}.
\end{proof}

\section{The isomorphism algorithm}\label{sec:iso:alg}

\begin{lem}
 \label{lem:structure-cliques}
 Let $G = (D,E)$ be a graph of tree width at most $k$ and let $H$ be a rooted
 graph such that $D \subseteq V(H)$.
 Then, one can compute an isomorphism-invariant rooted graph $H'$ such that
 \begin{enumerate}
  \item $H$ is an induced subgraph of $H'$,
  \item $\fdeg_{H'}(w) \leq \max\{d,k+1\} + 1$ for all $w \in V(H')$ where $d = \max_{w \in V(H)}\fdeg_H(w)$,
  \item $\depth(H') \leq \depth(H) + k + 2$, and
  \item for every clique $C \subseteq D$ there is a corresponding vertex $C \in
  V(H')$
 \end{enumerate}
 in time $2^{\mathcal{O}(k)} \cdot |V(H)|^{\CO(1)}$.
\end{lem}

Here, isomorphism invariant means that every isomorphism $\phi \in \iso(H_1,H_2)$, which naturally restricts to an isomorphism from $G_1$ to $G_2$,
can be extended to an isomorphism from $H_1'$ to $H_2'$.

\begin{proof}
 Let $G^{(0)} = G$ and for $i > 0$ let $G^{(i)} = G^{(i-1)}[V^{(i)}]$ where $V^{(i)} = \{v \in V(G^{(i-1)}) \mid \deg_{G^{(i-1)}}(v) > k\}$.
 Since $G$ has tree width at most $k$ it follows that there is some $i^{*} > 0$
 such that $G^{(i^{*})}$ is the empty graph.
 For every clique $C \subseteq D$ we define $i(C)$ to be the maximal $i \in \mathbb{N}$ such that $C \subseteq V(G^{(i)})$.
 Moreover, let $a(C) = C \setminus V(G^{(i(C)+1)})$.
 Observe that $a(C) \neq \emptyset$ and for every $v \in a(C)$ it holds $\deg_{G^{(i(C))}}(v) \leq k$ and $C \subseteq N_{G^{(i(C))}}[v]$.
 
 For $v \in D$ let $i(v)$ be the maximal $i \in \mathbb{N}$ such that $v \in V(G^{(i)})$.
 Let $A(v) = N_{G^{(i(v))}}[v]$.
 Note that $|A(v)| \leq k+1$.
 
 For a set $S$ we define the rooted graph $L_S$ to be the graph associated with the subset lattice of $S$, that is, $L_S = (\Pow(S),E_S,\emptyset)$ where $AB \in E_S$ if $A \subseteq B$ and $|B \setminus A| = 1$.
 Now let $H'$ be the following graph.
 For every $v \in D$ we attach the graph $L_{A(v)}$ to the vertex $v \in V(H)$, that is, the root node $\emptyset \in V(L_{A(v)})$ is connected to the vertex $v$.
 It can be easily verified that $H'$ satisfies all requirements.
 Observe that $C \subseteq A(v)$ for every $v \in a(C)$ where $C$ is a clique of $G$.
\end{proof}

\begin{theo}\label{theo:main-alg_iso}
Let $k\in\NN$ and let $G_1,G_2$ be connected graphs.
There is an algorithm that
either correctly concludes that $\tw(G_1) > k$,
or computes the set of isomorphisms $\iso(G_1,G_2)$
in time $2^{\CO(k\log^c k)}|V (G)|^{\CO(1)}$.
\end{theo}

\begin{proof}
We describe a dynamic programming algorithm $\CA$ that has
an input $(I_1,I_2)$ where $I_j=(G_j,S_j,T_j,\beta_j,\eta_j)$
where $G_j$ is a graph, 
$(T_j,\beta_j)$ is a tree decomposition
of the graph, $S_j\subseteq V(G_j)$ is a subset of
the vertices contained in the root bag, and $\eta_j$ is
a (partial) function that assigns nodes $t_j\in V(T_j)$ a
graph $\eta(t_j)$ which is isomorphism invariant w.r.t 
$T_{j,t_j},\beta_{j,t_j},G_{j,t_j}$ and $S_{j,t_j}$ as in \cref{rem:sub}.
The isomorphisms $\iso(I_1,I_2)$ are all isomorphisms $\phi:V(G_1)\to V(G_2)$
from $(G_1,S_1)$ to $(G_2,S_2)$ for which there is a bijection
between the nodes of the trees $\pi:V(T_1)\to V(T_2)$ such
that
for all $t\in V(T_1)$ we have that $\phi(\beta_1(t))=\beta_2(\pi(t))$.
The algorithm computes a coset
$\CA(I_1,I_2)=\iso(I_1,I_2)$.
In our recursive procedure we maintain the following properties of
the tree decomposition for each unlabeled bag $\beta_j(t_j)$
($\eta_j(t_j)=\bot$).\\
\begin{minipage}[t]{0.45\textwidth}
\begin{enumerate}[label=(U\arabic*)]
\item\label{eee:1} The graph $G_j^k[\beta_j(t_j)]$ is clique-separator free, and
\end{enumerate}
\end{minipage}
\begin{minipage}[t]{0.5\textwidth}
\begin{enumerate}[label=(U\arabic*)]
\setcounter{enumi}{1}
\item\label{eee:2} Each adhesion set of $t_j$ and also $S$ are cliques in $G_j^k$.
\end{enumerate}
\end{minipage}\\[1.5ex]

And for each labeled bag $\beta_j(t_j)$ we require the following.\\
\begin{minipage}[t]{0.45\textwidth}
\begin{enumerate}[label=(L\arabic*)]
\item\label{iii:2} The cardinality $|\beta_j(t_j)|$ is bounded by $\funclarge$,
\item\label{iii:6} $\eta_j(t_j)=H$ is a connected rooted graph such that
$\beta_j(t_j)\cup\beta_j(t_j)^2\subseteq V(H)$ and for each adhesion set $S$
there is a corresponding vertex $S\in V(H)$,
$\depth(H)\in \CO(k)$ and $\fdeg(v)\in k^{\CO(1)}$ for all $v\in H$,
\end{enumerate}
\end{minipage}
\begin{minipage}[t]{0.5\textwidth}
\begin{enumerate}[label=(L\arabic*)]
\setcounter{enumi}{2}
\item\label{iii:7} For each bag~$\beta_j(t_j)$ the adhesion sets of the children
are all equal to~$\beta_j(t_j)$ or the adhesion sets of the children are all
distinct,
\item\label{iii:8} If the adhesion sets are all equal and $\beta_j(t_j)$ is no
clique in $G_j^k$, then the cardinality $|\beta_j(t_j)|$ is bounded by
$\funcmedium$ and
the number of children of $t_j$ is bounded by $k$.
\end{enumerate}
\end{minipage}\\[1.5ex]

The initial input of $\CA$ consists of
the graphs $G_1,G_2$ together with their
canonical (unlabeled) clique separator decompositions $(T_j,\beta_j)$
of their $k$-improved graphs $G_j^k$ from~\cref{theo:cliqueDecomposition}.
For $S_j$ we choose the empty set.
Since the clique separator decomposition is canonical we have
$\iso(I_1,I_2)=\iso(G_1,G_2)$ and therefore
the algorithm computes all isomorphisms
between $G_1$ and $G_2$.\\
\underline{Description of $\CA$:}

For $j=1,2$ let $t_{j1},\ldots,t_{j\ell}$ be the children of $t_j$
and let $(T_{ji},\beta_{ji},\eta_{ji})$ be the
decomposition of the subtree rooted at $t_{ji}$.
Let $G_{ji}$ be the graph corresponding to $(T_{ji},\beta_{ji},\eta_{ji})$.
Let $V_{ji}:=V(G_{ji})$ and
let $S_{ji}\coloneqq \beta_j(t_j)\cap V_{ji}$ be the adhesion sets of the
children, let $Z_{ji}$ be $V_{ji}\setminus S_{ji}$,
and let $I_{ji}:=(G_{ji},S_{ji},T_{ji},\beta_{ji},\eta_{ji})$.
We assume that the isomorphisms
$\CA(I_{1i},I_{2i'})$ have already been
computed via dynamic programming.

We have to consider two cases depending on the root $t_j\in T_j$.

\begin{cs}
\case{$\beta_j(t_j)$ are both unlabeled}
Let $v_1$ be a vertex in $\beta_1(t_1)$
of degree at most
$k$ in $G_1^k[\beta_1(t_1)]$ and
compute the canonical
labeled tree
decomposition $(T_1',\beta_1',\eta_1')$ of $(G_1^k[\beta_1(t_1)],v_1)$ from
\cref{theo:labelDecomposition}.
For each vertex $v_2\in\beta_2(t_2)$
of degree at most $k$
in $G_2^k[\beta_2(t_2)]$ we compute the canonical
labeled tree
decomposition $(T_2',\beta_2',\eta_2')$ of $(G_2^k[\beta(t_2)],v_2)$ from
\cref{theo:labelDecomposition}.
Notice that the bags of $(T_j',\beta_j',\eta_j')$
satisfy \ref{iii:2}--\ref{iii:8}, because in the case $\tw(G_j)\leq k$,
the $k$-improvement does not increase the tree width
as seen in \cref{lem:imp}(2).
Notice that the adhesion sets $S_{ji}$ and $S_j$
are cliques in $G_j^k$ and these cliques must be
completely contained in one bag.
We attach the children $(T_{ji},\beta_{ji},\eta_{ji})$ to
$(T_j',\beta_j',\eta_j')$ by adding them to the highest possible bag
and
we choose a new root as the highest possible node $r_j\in V(T_j')$ 
such that $S_j\subseteq \beta_j'(r_j)$. 
By doing this, we obtain $(T_j'',\beta_j'',\eta_j'')$.
We need to recompute
$\eta_j''(s_j)$ for all nodes $s_j\in V(T_j)$ where new children
are attached to preserve Property \ref{iii:6}.
We call the algorithm in \cref{lem:structure-cliques}
with the input $G_j[\beta_j(s_j)],\eta_j''(s_j)$
and obtain $\eta_j'''(s_j)$.
Finally, we obtain $(T_j'',\beta_j'',\eta_j''')$
and define $I_j':=(G,S,T_j'',\beta_j'',\eta_j''')$.
We show that \ref{iii:2}-\ref{iii:8} remains satisfied.
By possibly introducing new bags we can preserve Property \ref{iii:7}:
The adhesion sets are either
pairwise different or all equal.
We can preserve Property \ref{iii:8}.
Consider a bag $s_j$ in which all adhesion sets are equal and
where $\beta_j''(s_j)$ is not a clique in $G_j^k$.
Then the children are not any of the attached children
$t_{j1},\ldots,t_{j\ell}$
and therefore the number of children of $s_j$
remains bounded by $\funcmedium$.

For each of the vertices $v_2 \in \beta_2(t_2)$ we compute the isomorphisms
$\CA(I_1',I_2'(v_{2}))$ recursively.
We return the smallest coset that contains the union
$\bigcup_{v_2 \in \beta_2(t_2)} \iso(I_1',I_2'(v_{2}))$.

\case{$\beta_j(t_j)$ are both labeled}
We consider two cases depending whether
the adhesion sets of $t_j$ are all different or not.

If all adhesion sets are different,
we use the algorithm from \cref{lem:big_iso}
as follows.
We encode the isomorphism problem $\iso(I_1,I_2)$
as an instance of coset-hypergraph-isomorphism.
First, we encode the tree decomposition as
\[\CI=(\beta_1(t_1),\beta_2(t_2),\{S_{1i}\},\{S_{2i}\},\chi_1,\chi_2,\CF,\Cf)\]
and $H_1=\eta_1(t_1),H_2=\eta_2(t_2)$
where $\CF,\Cf$ are the isomorphisms between the children
of $t_1$ and the children of $t_2$
restricted to their adhesion sets.
We can easily construct $\CF,\Cf$
with the isomorphisms $\CA(I_{1i},I_{2i'})|_{S_{1i}}$ that have already been computed.
The colorings $\chi_i$ are defined as follows.
Two $S_{1i}$ and $S_{2i'}$ get assigned the same color (i.e.\ $\chi_1(S_{1i}) = \chi_2(S_{2i'})$)
if and only if the set of isomorphisms $\CA(I_{1i},I_{2i'})$ is not empty.
Since $\CF,\Cf$ consist of isomorphisms
the consistency property \ref{eq:n-property}
is satisfied.
Next, we encode the edge relation of edges contained in the root
bag as 
$\CI'=(\beta_1(t_1),\beta_2(t_2),E(\beta(t_1))\cup\{S_1\},
E(\beta(t_2))\cup\{S_2\},\chi_1',\chi_2',\CF',\Cf')$.
To construct $\CF',\Cf'$ we define
for all edges
$e_1\in E(\beta_1(t_1)),e_2\in E(\beta_2(t_2))$
the allowed isomorphisms
$\Theta_{e_1}\tau_{e_1,e_2}$
as the set of bijections between $\beta_1(t_1)$ and $\beta_2(t_1)$
that map $e_1$ to $e_2$.
Moreover, the
allowed isomorphisms
$\Theta_{S_1}\tau_{S_1,S_2}$
is the set 
of bijections between $\beta_1(t_1)$ and $\beta_2(t_1)$
that map $S_1$ to $S_2$.
We
define $\chi_1,\chi_2$ as a coloring mapping
all edges to the integer 0
and mapping $S_1$ and $S_2$ to the integer 1.
Finally, we combine both instances
$\CI$ and $\CI'$ to obtain $\CI''$
and compute the isomorphisms $\iso(\CI'')$
by the algorithm in \cref{lem:big_iso}.
We extend the isomorphisms to the whole graph and
define and return $\hat\iso(\CI''):=\{\phi:V_1\to V_2\mid\phi|_{\beta_1(t_1)}
\in \iso(\CI''),\exists\pi\in\sym(\ell):\phi|_{V_{1i}}\in
\iso(I_{1i},I_{2\pi(i)})\}$.

If all adhesions sets are equal,
the procedure is analogous, but
we compute $\iso(\CI'')$ by calling
the algorithm from \cref{lem:small_iso}.
\end{cs}

\begin{correctness}
We consider the first case.
Since each isomorphism from $I_1$ to $I_2$
maps the vertex $v_1$ to some vertex $v_{2} \in \beta_2(t_2)$
we conclude that $\iso(I_1,I_2)= \bigcup_{v_2\in\beta_2(t_2)}\iso(I_1',I_2'(v_{2}))$.

We consider the second case.
We show that $\iso(I_1,I_2)=\hat\iso(\CI'')$.
Let $\phi\in\iso(I_1,I_2)$ be an
isomorphism, i.e.\ $\phi$
is an isomorphism from $(G_1,S_1)$ to $(G_2,S_2)$
for which there is a bijection
between the nodes of the trees $\pi:V(T_1)\to V(T_2)$ such
that for all $t\in V(T_1)$ we have that $\phi(\beta_1(t))=\beta_2(\pi(t))$.
Equivalently, since each $\pi$ maps the root of $T_1$ to the root
of $T_2$ we have that
$\phi|_{\beta_1(t_1)}\in\iso(G_1[\beta_1(t_1)],S_1,G_2[\beta_2(t_2)],S_2)$
and there is a bijection $\pi:\sym(\ell)$ between the children of $t_1$
and $t_2$ such that $\phi|_{V(G_{1i})}\in\iso(I_{1i},I_{2\pi(i)})$.
The condition
$\phi|_{\beta_1(t_1)}\in\iso(G_1[\beta_1(t_1)],S_1,G_2[\beta_2(t_2)],S_2)$
is equivalent to $\phi|_{\beta_1(t_1)}\in\iso(\CI')$.
Therefore, $\phi\in\iso(I_1,I_2)$, if and only if $\phi\in\hat\iso(\CI'')$.
\end{correctness}

\begin{runningtime}
Due to dynamic programming the number of calls is polynomial in $|V(G)|$.
More precisely, for each bag of the clique separator decomposition,
we construct a labeled tree decomposition and both have polynomially
many bags.

Consider the first case.
The construction
of the labeled tree decomposition from \cref{theo:labelDecomposition}
takes time $2^{\CO(k\log k)}|V(G)|^{\CO(1)}$.
Updating the graph $H=\eta_j''(s_j)$ in the case where the clique separators
where attached is done by \cref{lem:structure-cliques} in time
$2^{\CO(k\log k)}|V(H)|^{\CO(1)}$.
Notice that the depth of $H$
is bounded by $\CO(k)$
and the maximum forward-degree is bounded by $k^{\CO(1)}$.
Therefore this
gives a bound on the number of vertices
$|V(H)|\in 2^{\CO(k\log k)}$.

Consider the second case in which all adhesion sets are
different.
The isomorphism subroutine in
\cref{lem:big_iso} runs in time
$2^{\CO(k \cdot \log^{c}k)}$.
Notice that for given $\iso(\CI'')$ and $\iso(I_{1i},I_{2i'})$
the coset $\hat\iso(\CI'')$ can be constructed 
in polynomial time, since for each isomorphism $\phi''\in\iso(\CI'')$
there is already a unique $\pi\in\sym(\ell)$
such that $\phi''|_{S_{1i}}\in\iso(I_{1i},I_{2\pi(i)})|_{S_{1i}}$
and therefore it suffices to extend $\phi$ using generators
for $\iso(I_{1i},I_{2\pi(i)})$.

Next, consider the second case in which all adhesion sets are equal.
Assume that $\beta_j(t_j)$ is not a clique in $G_j^k$.
Then by Property \ref{iii:8}
the cardinality $|\beta_j(t_j)|$ is bounded by $\funcmedium$
and the number of children of
$t_j$ is bounded by $k$.
In this case we have that $|\CF|\leq k$
for our subroutine which leads to a run time
$|\CF|!^{\CO(1)}2^{\CO(\log^c(|\beta_1(t_1)|))}\subseteq 2^{\CO(k\log k)}$.
In the case where $\beta_j(t_j)$ is a clique in $G_j^k$
we have that $|\beta_j(t_j)|\leq k+1$ and
$|\CF|\leq |V(G_j)|$.
This leads to a run time of
$(|V_1|!|\CF|)^{\CO(1)}\subseteq 2^{\CO(k\log k)}|V(G_1)|^{\CO(1)}$.
It remains to explain how the coset $\hat\iso(\CI'')$ can be constructed 
in polynomial time.
The cosets are colored according to the isomorphism types of
the graphs $G_{ji}$ which ensures property \ref{eq:n-property-small}.
Let $P=P_1\cupdot\ldots\cupdot P_p$ be a partition of $[\ell]$
into isomorphism classes, i.e. $i,j\in P_s$ if and only if
$\Theta_i\tau_{i,k}=\Theta_j\tau_{j,k}$ for some (and due to property
\ref{eq:n-property-small} for all) $k\in[\ell]$.
Now, for each isomorphism $\phi''\in\iso(\CI'')$
there is a unique coset $\Pi$ of $\sym(P_1)\times\ldots\times\sym(P_p)$
such that $\phi''\in\iso(I_{1i},I_{2\pi(i)})|_{\beta_1(t_1)}$
for all $\pi\in\Pi$
and therefore it suffices to extend $\phi$ using generators
for $\iso(I_{1i},I_{2\pi(i)})$ and $\Pi$.
\end{runningtime}
\end{proof}

\section{Canonization}\label{sec:canoniztion:tools}

In this section we adapt our techniques to obtain an algorithm which computes a
canonization for a given graph of tree width at most $k$.
Since Babai's quasi-polynomial time algorithm \cite{DBLP:conf/stoc/Babai16} only
tests isomorphism of two given input graphs and can not be used for
canonization purposes, our canonization algorithm has a slightly worse running
time.
However, our canonization algorithm still significantly improves on the previous best due to Lokshtanov, Pilipczuk, Pilipczuk, and Saurabh \cite{lokshtanov2017fixed}.

One of the main tasks for building the canonization algorithm out of our isomorphism test is to adapt Lemma \ref{lem:big_iso}.
To achieve this we shall use several group theoretic algorithms concerned with canonization.
We start by giving the necessary background.

\subsection{Background on canonization with labeling cosets}

Two graphs graphs $G$ and $H$ are called isomorphic with respect to a
group
$\Delta\leq\sym(V(G))$, if there is a bijection $\delta\in\Delta$ such
that $G^\delta=H$. In this case we write $G\cong_\Delta H$.
By $\aut_\Delta(G)$ we denote the automorphism group
of a graph $G$ restricted to $\Delta$, i.e.,
all $\delta\in\Delta$ such that $G^\delta=G$.

In the following let $\CX$ denote a class
of graphs, closed under isomorphisms.
\begin{defn}
Assume that $V = \{1,\ldots,|V|\}$.
A function $\cf_{\Delta}:\CX\to\CX$ is a \emph{canonical form} with
respect to a group $\Delta\leq\sym(V)$ for a graph
class $\CX$, if
$\cf_\Delta(G)\cong_{\Delta} G$ for all $G\in\CX$, and
if $G\cong_{\Delta} H$ implies $\cf_\Delta(G)=\cf_\Delta(H)$ for all
  $G,H\in\CX$.
\end{defn}

\begin{defn}
Assume that $V = \{1,\ldots,|V|\}$.
The \emph{canonical labelings} of a graph~$G$ with respect to a group $\Delta\leq\sym(V(G))$ are the
elements in $\Delta$ that bring $G$ in the canonical form, i.e.,
$\can(G;\Delta)\coloneqq\{\delta\in\Delta~|~G^\delta=\cf_\Delta(G)\}$.
\end{defn}
For the purpose of recursion, it is usually more convenient
to compute the entire coset of canonical labelings
rather than just a canonical form.

\begin{lem}[\cite{babai1983canonical}]\label{lem:char}
Assume that $V = \{1,\ldots,|V|\}$.
The following three conditions characterize canonical labelings.
\begin{enumerate}[label=\textnormal{(CL\arabic*)},leftmargin=1.5cm]
  \item\label{luks:cl1} $\can(G;\Delta)\subseteq\Delta$,
  \item\label{luks:cl2} $\can(G;\Delta)=\delta\can(G^\delta;\Delta)$
  for all $\delta\in\Delta$, and
  \item\label{luks:cl3} $\can(G;\Delta)=\aut_{\Delta}(G)\pi$
  for some (and thus for all) $\pi\in\can(G;\Delta)$.
\end{enumerate} 
\end{lem}
These conditions imply that the assignment of~$G^\pi$ for an arbitrary~$\pi
\in \can(G;\Delta)$ is independent of the choice of~$\pi$ and
that~$G^\pi$ is a canonical form of~$G$
with respect to $\Delta$, justifying the name.

Above we made the assumption that
$V$ is a fixed linearly ordered set.
We can drop this assumption
and can define canonical labelings
for arbitrary vertex sets as seen next.
For a labeling coset $\tau\Delta\leq\lab(V(G))$ we define
$\can(G;\tau\Delta)$ to be the set~$\tau\can(G^\tau;\Delta)$ which is well defined
due to Property~\ref{luks:cl2}.
In general, a canonical form $\cf$ for a graph $G$
over an arbitrary (not necessary ordered) vertex set $V$
is a function such that
$\cf(G)\cong G$ for all $G\in\CX$, and
$G\cong H$ implies $\cf(G)=\cf(H)$ for all
$G,H\in\CX$.
Notice that $G^\pi$ defines such a canonical form for
some $\pi\in\can(G;\lab(V(G)))$.
An algorithm to compute $\can(G;\lab(V(G)))$ for
a graph of tree width at most $k$ will be present in
the last section.

The \emph{composition-width} of a coset~$\tau\Delta$,
denoted~$\cw(\tau\Delta)$ is the smallest integer~$k\in \NN$ such that
$\tau\Delta\in\tilde\Gamma_k$, that is, 
every non-Abelian composition factor of~$\Delta$ is isomorphic to a subgroup
of~$\sym(k)$.
The composition width of a collection of cosets~$\cw(\CS)$ is
the maximum composition width among all elements of~$\CS$.
Let $\omega(n)$ be the smallest function such that,
if $\Delta\leq\sym(X)$ is primitive, then
$|\Delta|\leq|X|^{\omega(\cw(\Delta))}$.

\begin{theo}[\cite{babai1983canonical}]\label{lem:luks}
Given a graph $G$ and a labeling coset $\tau\Delta\in\lab(V)$,
a canonical labeling coset $\can(G;\tau\Delta)$
can be computed in time $|V(G)|^{\CO(\omega(\cw(\Delta)))}$.
\end{theo}

It is actually known that $\omega(n)\in\CO(n)$ as stated in
\cite{DBLP:conf/focs/BabaiKL83}, see \cite{liebeck1999simple}.
Analogously, for a hypergraph $X$
we first define the canonical labelings with respect
to a group and a start segment of $\NN$ as the vertex set.
Meaning $\can(X;\Delta)$ satisfies \ref{luks:cl1}-\ref{luks:cl3}
with $X$ in place of $G$.
Then for graphs on arbitrary vertex set and a given labeling
coset we define $\can(X;\rho\Delta)$ to be the set
$\rho\can(X^\rho;\Delta)$.

\begin{theo}[\cite{miller1983isomorphism},
see \cite{Ponomarenko}]\label{theo:miller}
Given a hypergraph $X = (V,E)$ and a labeling coset
$\tau\Delta\in\lab(V)$,
a canonical labeling coset $\can(X;\tau\Delta)$
can be computed in time $(|V(X)||E(X)|)^{\CO(\omega(\cw(\Delta)))}$.
\end{theo}

We describe several polynomial time operations that can used in general when dealing with labeling cosets of a set of natural numbers.

We first define an ordering
on~$\lab(\{1,\ldots,n\})= \sym(\{1,\ldots,n\})$.
For two such permutations~$\pi_1,\pi_2$ define~$\pi_1\prec \pi_2$ if there is an~$i\in \{1,\ldots,n\}$ such that~$\pi_1(i) <\pi_2(i)$ and~$\pi_1(j)=\pi_2(j)$ for all~$1\leq j<i$.
We extend the definition to labeling cosets of~$\{1,\ldots,n \}$ as follows.
Suppose~$\tau_1\Theta_1,\tau_2\Theta_2\leq \lab(\{1,\ldots,n\})$ are labeling cosets.
Then we define~$\tau_1\Theta_1\prec\tau_2\Theta_2$ if~$|\tau_1\Theta_1|
<|\tau_2\Theta_2|$ or if $|\tau_1\Theta_1|=|\tau_2\Theta_2|$ and the smallest element of~$\tau_1\Theta_1\setminus
\tau_2\Theta_2$ is smaller (w.r.t.~to~$\prec$ on permutations) than the smallest element of~$\tau_2\Theta_2\setminus \tau_1\Theta_1$.

Let~$\tau_1,\tau_2,\ldots,\tau_t$ be a set of bijections with the same domain
and range. We denote by~
\[\lcosetgen \tau_1,\tau_2,\ldots,\tau_t
  \rcosetgen\]
the smallest coset containing all~$\tau_i$. This coset can be computed in polynomial
time (see for example Lemma 16 of the arXiv
version of~\cite{DBLP:conf/focs/GroheS15}).

\begin{lem}\label{lem:can:sims}
There is a polynomial time algorithm that, given a labeling coset~$\tau\Theta$ of the set~$\{1,\ldots,n \}$, as~$\tau$ and a generating set~$S$ of~$\Theta$, computes a sequence of at most~$n\log n$ elements~$(\tau_1,\tau_2,\ldots,\tau_t)$ of~$\tau\Theta$ such that
\begin{enumerate}
\item $\tau\Theta = \lcosetgen \tau_1,\ldots,\tau_t\rcosetgen$,
\item $\tau_i$ is the smallest element of~$\tau\Theta$ not contained in~$\lcosetgen\tau_1,\ldots,\tau_{i-1}\rcosetgen$, and
\item the output of the algorithm only depends on~$\tau\Theta$ (and not on~$\tau$ or~$S$).
\end{enumerate}
\end{lem}

\begin{proof}
It suffices to find a solution for the following task:
Given~$\tau_1,\tau_2,\ldots,\tau_\ell$ find the smallest element
in~$\tau\Theta$ not contained $\bar\langle \tau_1,\tau_2,\ldots,\tau_\ell\bar\rangle$.

This task can be solved in polynomial time as follows.
We find the smallest~$j$ such that not all elements of~$\tau\Theta$ have the
same image of~$j$.
Otherwise, if each element of $\tau\Theta$ has the same image,
then $\tau\Theta$ consists of only one element and we are done.
With standard group theoretic techniques, including the
Schreier-Sims algorithm, we can compute the possible images of~$j$ under~$\tau\Theta$, say~$i_1,i_2,\ldots,i_k$.
For each~$m\in \{1,\ldots,k\}$ we compute~$(\tau\Theta)_{j\mapsto i_{m}}$, the
set of those elements in~$\tau\Theta$ that map~$j$ to~$i_{m}$.
Notice that if $\tau\Theta$ contains an element not in
$\bar\langle \tau_1,\tau_2,\ldots,\tau_\ell\bar\rangle$,
then there is also a $(\tau\Theta)_{j\mapsto i_{m}}$
that contains an element not in $\bar\langle \tau_1,\tau_2,\ldots,\tau_\ell\bar\rangle$.
We find the smallest~$i_m$ such that~$(\tau\Theta)_{j\mapsto i_{m}}$ is
not contained in~$\bar\langle \tau_1,\tau_2,\ldots,\tau_\ell\bar\rangle$.
Now we compute~$(\bar\langle
\tau_1,\tau_2,\ldots,\tau_\ell\bar\rangle)_{j\mapsto i_{m}}$ and recurse
with~$(\tau\Theta)_{j\mapsto i_{m}}$.

Each of the operations can be performed in polynomial time  with standard techniques
(see~\cite[Beginning of Section 9]{DBLP:conf/focs/GroheS15}, for a detailed 
exposition, the reader can refer to an arXiv version of the paper).
\end{proof}

Note that the lemma can be used to compute canonical generating sets of a permutation group defined on a linearly ordered set.
\begin{cor}\label{cor:order:labelings:of:natural:numbers}
\begin{enumerate}
\item Given a labeling coset~$\tau\Theta$ of~$\{1,\ldots,n \}$  we can compute the smallest element (w.r.t.~``$\prec$'') in time polynomial in~$n$.
\item Given a collection of labeling cosets~$\{\tau_1\Theta_1, \ldots,\tau_t\Theta_t\}$ of~$\{1,\ldots,n \}$ via generating sets, we can compute their order (w.r.t.~``$\prec$'') in time polynomial in~$t$ and~$n$.
\end{enumerate}
\end{cor}


\subsection{Graphs with coset-labeled hyperedges}

With the results establiched in the previous subsection we can now give a variant of Lemma \ref{lem:big_iso} suitable for canonization.
For this purpose we adapt the problem of coset-hypergraph isomorphism to the context of computing canonical labelings.
More precisely, we define a new combinatorial object, namely hypergraphs with coset-labeled
hyperedges, and canonical labelings for those objects.

A \emph{hypergraph with coset-labeled hyperedges} is a 4-tuple~$(V,\CS,\chi,\CF)$ consisting of
\begin{enumerate}
\item a vertex set $V$,
\item the hyperedges~$\CS\subseteq \Pow(V)$ which form a collection of subsets of~$V$, 
\item a coloring $\chi\colon \CS\rightarrow \NN$ of the hyperedges, and  
\item a collection~$\CF =\{\tau_S\Theta_S\leq\lab(S)\mid S\in \CS \}$ of
labeling cosets of the hyperedges (that is, for each~$S$ the collection contains a coset~$\tau_S\Theta_S$)
which satisfies~$\Theta_S = \Theta_{S'}$ for all~$S,S'\in \CS$
with~$\chi(S)= \chi(S')$).
\end{enumerate}

For a map~$\phi:V\to V'$ we
define~$(V,\CS,\chi,\CF)^\phi$ to be the tuple~$(V^\phi, 
\CS^\phi, \chi^\phi,\CF^\phi)$, where~$\CS^\phi = \{S^\phi\mid
S\in \CS\}$,~$\chi^\phi=\phi^{-1}\chi$ and~$\CF^\phi$ is the
collection~$\{\phi^{-1} \tau_S\Theta_S\}$.
The automorphism group of a
hypergraph with coset-labeled hyperedges $X$
with respect to a group $\Delta\leq\sym(V(X))$, written $\aut_\Delta(X)$,
are all permutations $\delta\in\Delta$ such that 
$X^\delta = X$.

Now, we adapt the canonization to our combinatorial object and 
define \emph{canonical labeling} for hypergraphs  with coset-labeled
hyperedges.
As before, we assume that $V(X)=\{1,\ldots,|V(X)|\}$.
A \emph{canonical labeling}
is a map
that assigns every
pair consisting of a
hypergraph with coset-labeled hyperedges~$X$
and a group $\Delta\leq\sym(|V|)$
a labeling coset $\can(X;\Delta)$ such
that
\begin{enumerate}[label=(CL\arabic*),align=left,leftmargin=1.25cm]
\item\label{hyper:cl1} $\can(X;\Delta)\subseteq \Delta$,
\item\label{hyper:cl2} $ \can(X;\Delta) = \delta \can(X^\delta;\Delta)
$ for all $\delta\in\Delta$, and
\item\label{hyper:cl3} $\can(X;\Delta) =\aut_\Delta(X) \pi $ for some (and thus for all)~$\pi\in \can(X;\Delta)$.
\end{enumerate}

As before, we extend the canonical labelings to arbitrary (not necessary
ordered) vertex sets.
For a labeling coset $\tau\Delta\leq\lab(V)$
we write $\can(X;\tau\Delta)$ for $\tau\can(X^\tau;\Delta)$.
This is well defined due to Property~\ref{hyper:cl2}. 

We are now ready to formulate the analogue of Lemma \ref{lem:big_iso} for the canonization algorithm.

\begin{lem}\label{theo:big}
There is an algorithm that computes a canonical labeling of a
pair~$(X;\tau\Delta)$ consisting of a hypergraph with coset-labeled
hyperedges~$X= (V,\CS,\chi,\CF)$ and a labeling coset $\tau\Delta\leq\lab(V)$
which has a running time
of~$(|V(X)||\CS|)^{\CO(\max\{\cw(\Delta),\cw(\CF)\})}$.
\end{lem}

\begin{proof}
Since $\can(X;\tau\Delta)$
is $\tau\can(X^\tau;\Delta)$ by definition,
we need to give an algorithm for given
a pair~$(X; \Delta)$ with~$X=
(V,\CS,\chi,\CF)$ that computes $\can(X;\Delta)$.
If $\Delta$ does not preserve $\CS$ and $\chi$,
we proceed as follows.
Using \cref{theo:miller}
we compute a
canonical labeling coset~$\tau\Delta' = \can((V,\CS,\chi);\Delta)$
of the hypergraph so that~$\Delta'$ preserves~$\CS^\tau$ and
$\chi^\tau$.
We compute and return $\can(X;\tau\Delta')$ recursively
(Notice that $\can(X;\tau\Delta')$ is defined
as $\tau\can(X^\tau;\Delta')$).

So we can assume in the following that
$\Delta$ preserves $\CS$ and $\chi$, i.e,
$\CS^\delta=\CS$ and $\chi^\delta=\chi$ for all $\delta\in\Delta$.
Let~$U$ be the set of pairs~$\{(S,s)\mid S\in \CS,s\in S\}$.
We can think of~$U$ as the
disjoint union of all sets~$S\in \CS$. For notational convenience we
define~$U_S \coloneqq \{S\}\times S$.

Next we define a labeling~$\hat\Delta$ on the disjoint union~$V \cup U$.
Roughly speaking, we define a labeling on $V \cup U$ that orders
$V$ according to a $\delta\in\Delta$ and orders the sets $U_S$
according to the ordering of the sets $S^\delta$.
Inside the set $U_S$ the vertices are ordered according
to some $\rho_S\in\tau_S\Theta_S$.
More formally, we
define~$\hat\Delta$ to be the labeling of~$V\cup U$ consisting of the
labelings~$\hat\delta$ for which there are
$\rho_S\in\tau_S\Theta_S$.
\begin{enumerate}
\item $\hat\delta|_V=\delta$ for some element of~$\delta\in \Delta$,
\item $(S,s)^{\hat\delta}<(S',s')^{\hat\delta}$ if $S^\delta\prec S'^\delta$
(the ordering $\prec$ of subsets of~$\NN$ as defined in the preliminaries), and
\item $(S,s)^{\hat\delta}<(S,s')^{\hat\delta}$
if $s^{\rho_S}<s'^{\rho_S}$.
\end{enumerate}

The coset~$\hat\Delta$ is a labeling coset of the set~$V\cup U$ and has
composition-width at most $\max\{\cw(\Delta),\cw(\CS)\}$.

Consider the bipartite graph~$Y$ defined on~$V\cup U$ by connecting~$v$ and
$(S,v)$ for each~$S\in \CS$ for which~$v\in S$.
Return the canonical labeling coset of the graph $Y$
restricted to $V$, i.e., $\can(Y;\hat\Delta)|_V$.

\begin{canonicallabeling}
We argue that~$\can(Y;\hat\Delta)$ restricted to~$V$
is a canonical labeling coset of~$(X;\Delta)$.

\ref{hyper:cl1} By definition, all elements of~$\hat\Delta$ restricted to~$V$
are in~$\Delta$.

\ref{hyper:cl2}
We consider the result of applying the algorithm
to~$(X^\delta;\Delta)$ for any $\delta\in\Delta$.
In the case that $\Delta$ does not preserve
$\CS$ and $\chi$, the algorithm returns
$\delta^{-1}\tau\Delta'$.
So assume that $\Delta$ preserves $\CS$ and $\chi$
already,
In this case
the definition of $V\cup U$ remains unchanged since $\Delta$ preserve $\CS$.
In place of~$\hat\Delta$, the algorithm computes~$\hat\Delta^\delta$. In
Place of~$Y$ the algorithm computes~$Y^\delta$.
But then the result of the computation
is~$\can(Y^\delta;\hat\Delta^\delta)|_V = \delta^{-1}\can(Y;\hat\Delta )|_V$.

\ref{hyper:cl3} Suppose~$\hat\delta ,\hat\delta' \in \can(Y;\hat\Delta)$
and $\hat\delta|_V=\delta$ and $\hat\delta'|_V=\delta'$.
We need to
show that~$\delta\delta'^{-1}$ is an automorphism of~$X$. Since~$\delta
,\delta \in \Delta$,
it is clear that~$\delta\delta'^{-1}$ preserves~$\CS$ and~$\chi$.
It remains to show that~$\delta\delta'^{-1}$ is compatible with~$\CF$.
In other words, we need to show that~$\CF^{\delta} = \{\delta^{-1}
\tau_S\Theta_S\mid S\in \CS \}$ and~$\CF^{\delta'} = \{{\delta'}^{-1}
\tau_S\Theta_S\mid S\in \CS \}$ are actually the same set.

Choose~$S\in \CS$. Let~$S'$ be the element of~$\CS$ for
which~$S^\delta=S'^{\delta'}$. Such an element exists and is unique.
Since $\chi$ is preserved we have~$\Theta_{S}=\Theta_{S'}$ and
therefore it suffices to show that~$\delta^{-1} \tau_S\Theta_S =\delta'^{-1}
\tau_{S'} \Theta_{S}$.

Consider the graph~$Y^{\hat\delta} = Y^{\hat\delta'}$. Let~$\gamma_S$ be the
map from~$S^{\hat\delta}$ to~$(U_{S})^{\hat\delta}$ that maps
each~$v^\delta\in S^\delta$ to~$(S,v)^{\hat\delta}$.
Similarly let~$\gamma'_{S'}$ be the map from~$S'$ to~$(U_{S'})^{\hat\delta'} =
(U_{S})^{\hat\delta}$ that maps each~$v^{\delta'}\in S'^{\delta'}$
to~$(S',v)^{\hat\delta'}$.
In the graph~$Y^{\hat\delta}$, for~$s\in S^\delta$ we
have that~$s^{\gamma_S}$ is the only vertex in~$(U_{S})^{\hat\delta}$
that is a neighbor of~$s$.
Similarly, in~$Y^{\hat\delta'} =Y^{\hat\delta}$,
we have that~$s^{\gamma'_{S'}}$ is the only vertex
in~$(U_{S'})^{\hat\delta'}$ that is a neighbor of~$s$. 
It follows that~$\gamma_{S} =\gamma'_{S'}$.

By construction of~$\hat\Delta$ we know that~$\delta\gamma_S=\rho_S+m$
for some~$\rho_S\in\tau_S\Theta_S$ and some integer $m$.
Similarly~$\delta'\gamma'_{S'}=\rho'_{S'}+m'$ for some $\rho'_{S'}\in
\tau'_{S'}\Theta_{S'}$ and some integer $m$.
In fact, it must be the case that~$m=m'$ since~$(U_{S'})^{\hat\delta'} =
(U_{S})^{\hat\delta}$.
We conclude that
$\delta\delta'^{-1}=\delta\gamma_S\gamma'^{-1}_{S'}\delta'^{-1}
=\rho_S\rho'^{-1}_{S'}$.
Therefore $\delta\delta'^{-1}\tau'_{S'}\Theta_S=\tau_{S}\Theta_S$ as desired.
\end{canonicallabeling}

\begin{runningtime}
We use two subroutines throughout the algorithm.
The first one is Miller's canonization algorithm
for hypergraphs from \cref{theo:miller} which runs in time
$(|V||S|)^{\CO(\cw(\Delta))}$.
The second one is a canonization algorithm for
the graph $Y$ which runs in time
$(|V|+|V||\CS|)^{\CO(\max\{\cw(\Delta),\cw(\CS)\})}$.
\end{runningtime}
\end{proof}

We also need a second combinatorial subproblem that handles the case when~$V$ is
small but the family~$\CF$ is possibly large (this is analogue to multiple-colored-coset-isomorphism and Lemma \ref{lem:small_iso}).

A \emph{set with multiple colored labeling cosets} is a tuple~$(V,\mathcal{F},\chi)$ consisting of 
\begin{enumerate}
\item a vertex set $V$,
\item a collection~$\mathcal{F} =\{\tau_1\Theta_1, \tau_2\Theta_2,\ldots,\tau_t\Theta_t\}$ of labeling cosets of~$V$, and
\item a coloring~$\chi \colon \{1,\ldots,t\} \rightarrow \mathbb{N}$ assigning to every labeling coset in~$\mathcal{F}$ a positive integer.
\end{enumerate}

For a map~$\phi:V\to V'$ we
define~$(V,\CF, \chi)^\phi$ analogously to before. Similarly the
automorphism and canonical labelings of~$(V,\CF, \chi)$ are defined
analogously to above.

We are now ready to formulate
a lemma that is the analogue of
\cref{lem:big_iso}, but with a
weaker run time
since Babai's quasi-polynomial time result
does not lead to canonization yet.

\begin{lem}\label{theo:small}
There is an algorithm that computes a canonical labeling of a
pair~$(X;\tau\Delta)$ consisting of a set with multiple colored labeling
cosets~$X= (V,\mathcal{F},\chi)$ and a labeling coset $\tau\Delta\leq\lab(V)$
which has a running time of~$(|V|!|\CF|)^{\CO(1)}$.
\end{lem}

\begin{proof}
Since $\can(X;\tau\Delta)$
is $\tau\can(X^\tau;\Delta)$ by definition,
we need to give an algorithm that
computes a canonical labeling for a given group
$\Delta\leq\sym(\{1,\ldots,|V|\})$
and a given set with multiple colored labeling cosets
on a ground set satisfying $V=\{1,\ldots,|V|\}$.
Let~$\mathcal{F}= \{\tau_1\Theta_1,  \tau_2\Theta_2, \ldots, \tau_t\Theta_t\}$ be the family of labeling cosets.

For each labeling~$\delta\in \Delta$ we do the following.
Consider the sets of the form~$\delta^{-1} \tau_i\Theta_i$.
We order these sets so that~$\delta^{-1} \tau_i\Theta_i\prec_\chi \delta^{-1}
\tau_j\Theta_j$ if~$\chi(\tau_i\Theta_i) <  \chi(\tau_j\Theta_j)$ or
if~$\chi(\tau_i\Theta_i) =  \chi(\tau_j\Theta_j)$ and~$\delta^{-1}\tau_i\Theta_i
\prec\delta^{-1}\tau_j\Theta_j$, where~$\prec$ is the ordering of coset labelings
of~$\{1,\ldots,n\}$ defined in \Cref{sec:canoniztion:tools}. By
Corollary~\ref{cor:order:labelings:of:natural:numbers} this ordering can be
computed in polynomial time. For every~$\delta$ we get an ordered
sequence~$\delta^{-1} \tau_{i_1}\Theta_{i_1},
\ldots,\delta^{-1} \tau_{i_t}\Theta_{i_t}$.

Let~$\{\delta_1,\ldots,\delta_s\}$ be the collection of those~$\delta$ for which this ordered sequence is lexicographically minimal.
Then we let~$\lcosetgen \delta_1,\ldots,\delta_s\rcosetgen$ be the output of the algorithm.

It follows from the construction that~$\lcosetgen \delta_1,\ldots,\delta_s\rcosetgen$ satisfies the properties of a canonical labeling coset.

By Corollary~\ref{cor:order:labelings:of:natural:numbers}, for each~$\delta$ the time requirement of the algorithm is polynomial in~$|\mathcal{F}|$ and~$|V|$. There are at most~$|V|!$ many elements in~$\Delta$, and thus at most equally many choices for~$\delta$, giving us the desired time bound.
\end{proof}

\subsection{The canonization algorithm}\label{sec:canon:algo}

Finally, we have assembled all the tools to describe an algorithm that canonizes
graphs of bounded tree width.

\begin{theo}\label{theo:main-alg}
Let $k\in\NN$ and let $G$ be a connected graph.
There is an algorithm that
either correctly concludes that $\tw(G) > k$,
or computes a canonical labeling coset $\can(G;\lab(V))$
in time $2^{\CO(k^2\log k)}|V (G)|^{\CO(1)}$.
\end{theo}

We remark that the algorithm is closely related to 
the one described in \cref{theo:main-alg_iso}.
In the isomorphism-algorithm
the graph $\eta(t)$ serves as a tool to exploit
the structure of each bag $t$.
As the isomorphism-techniques depend on Babai's quasi-polynomial
time result which does not extend to canonization yet, we will now use the
fact that additionally each bag can be guarded
with a labeling coset $\alpha(t)$ of bounded
composition-width. The
labeling coset $\alpha(t)$ serves a
tool for efficient
canonization.
The main change of the algorithm therefore arises in \ref{iii:6}.

\begin{proof}
We describe a dynamic programming algorithm $\CA$ with input $I=(G,S,T,\beta,\alpha)$ where
$(T,\beta)$ is a tree decomposition
of a graph
$G$ and $S\subseteq V(G)$ is a subset of
the vertices contained in the root bag, and $\alpha$ is a (partial)
function that maps nodes $t\in V(T)$ to 
labeling cosets $\alpha(t)\leq\lab(\beta(t))$.
The algorithm computes a
labeling coset $\CA(I)$ such
that
\begin{enumerate}[label=(CL\arabic*),leftmargin=1.5cm]
\item\label{algo:cl1} $\CA(I)\leq\lab(V(G))$ and $\CA(I)|_S\leq\lab(S)$, and
\item\label{algo:cl2} $\CA(I) = \tau
\CA(I^\tau)$ for all~$\tau\in\lab(V(G))$, and
\item\label{algo:cl3}  $ \aut(G,S,T,\beta,\alpha)\pi \subseteq
\CA(I)\subseteq\aut(G)\pi$ for some (and thus for all) $\pi\in\CA(I)$.
\end{enumerate}

In our recursive procedure we maintain the following properties of
the tree decomposition for each unlabeled bag $\beta(t)$
($\alpha(t)=\bot$).\\
\begin{minipage}[t]{0.45\textwidth}
\begin{enumerate}[label=(U\arabic*)]
\item\label{ee:1} The graph $G^k[\beta(t)]$ is clique-separator free, and
\end{enumerate}
\end{minipage}
\begin{minipage}[t]{0.5\textwidth}
\begin{enumerate}[label=(U\arabic*)]
\setcounter{enumi}{1}
\item\label{ee:2} Each adhesion set of $t$ and also $S$ are cliques in $G^k$.
\end{enumerate}
\end{minipage}\\[1.5ex]

For each labeled bag $\beta(t)$ we require the following.\\
\begin{minipage}[t]{0.45\textwidth}
\begin{enumerate}[label=(L\arabic*)]
\item\label{ii:2} The cardinality $|\beta(t)|$ is bounded by $\funclarge$,
\item\label{ii:4} The number of children of $t$ is bounded by
$k\funclarge^2+2^k\funclarge$
and $\alpha(t)\in\gammak$, and
\end{enumerate}
\end{minipage}
\begin{minipage}[t]{0.5\textwidth}
\begin{enumerate}[label=(L\arabic*)]
\setcounter{enumi}{2}
\item\label{ii:7} For each bag~$\beta(t)$ the adhesion sets of the children are
all equal to~$\beta(t)$ or the adhesion sets of the children are all distinct.
\item\label{ii:8} If the adhesion sets are all equal,
then the cardinality $|\beta(t)|$ is bounded by $\funcmedium$.
\end{enumerate}
\end{minipage}\\[1.5ex]

The initial input of $\CA$ is the
canonical clique separator decomposition $(T,\beta)$
of the $k$-improved graph $G^k$ from~\cref{theo:cliqueDecomposition}.
For $S$ we choose the empty set.
Since the clique separator decomposition is canonical we have
$\aut(G,S,T,\beta,\bot)=\aut(G)$ and therefore
the algorithm computes a canonical labeling coset of~$G$.\\
\underline{Description of $\CA$:}

Let $t_1,\ldots,t_{\ell}$ be the children of $t$ and
let $(T_i,\beta_i,\alpha_i)$ be the
decomposition of the subtree rooted at $t_i$.
Let $G_i$ be the graph corresponding to $(T_i,\beta_i,\alpha_i)$.
Let $V_i:=V(G_i)$,
let $S_i\coloneqq \beta(t)\cap V_i$ be the adhesion sets of the children, and
let $Z_i$ be $V_i\setminus S_i$.
We assume that the canonical labeling subcosets
$\tau_i\Theta_i\coloneqq\CA(G_i,S_i,T_i,\beta_i,\alpha_i)$ have already been
computed via dynamic programming.

We have to consider two cases depending on the root $t\in T$.
\begin{cs}

\case{$\beta(t)$ is unlabeled}
For each vertex $v\in\beta(t)$ of degree at most $k$
in $G^k[\beta(t)]$
we compute the canonical coset-labeled tree decomposition
$(T',\beta',\alpha')$ of $(G^k[\beta(t)],v)$ from
\cref{theo:labelDecomposition}.
Notice that the bags of $(T',\beta',\alpha')$
satisfy \ref{ii:2}-\ref{ii:8}, because in the case $\tw(G)\leq k$,
the $k$-improvement does not increase the tree width
as seen in \cref{lem:imp}(2).
Notice that the adhesion sets $S_i$ are cliques in $G^k$ and these cliques must
be completely contained in one bag.
We attach the children $(T_i,\beta_i,\alpha_i)$ to $(T',\beta',\alpha')$
adding them to the highest possible bag
and
we choose a new root as the highest possible node $r\in V(T')$ 
such that $S\subseteq \beta'(r)$.
By doing this, we obtain $(T'',\beta'',\alpha'')$.
We ensure Property \ref{ii:4} because
the number of children attached to one bag can be bounded by the number
of cliques in $G^k$ that contains the bag.
The number of cliques again can be bounded by $2^k\funclarge$
which is a consequence of \cref{lem:structure-cliques}.
By possibly introducing new bags we can preserve Property \ref{ii:7}:
The adhesion sets are either
pairwise different or all equal.
Since the adhesion sets $S_i$
are cliques in $G^k$ and therefore their
size is bounded by $k+1$, we also preserve
Property \ref{ii:8}.

For each $v$ we compute
$\Theta_v\coloneqq\CA(G,S,T'',\beta'',\alpha'')$ recursively.
For each~$v$ we compute the canonical form obtained by applying an arbitrary element of~$\Theta_v$ to the graph.
Let~$M$ be those~$v$ for which the obtained canonized graph is minimal with respect to lexicographic comparison of the adjacency matrix.
Then we return the smallest coset containing~$\bigcup_{v\in M} \Theta_v$.

\case{$\beta(t)$ is labeled}
We apply Luks' algorithm from \cref{lem:luks} to
get a subcoset that respects the edge relation and get
$\psi'\Delta'\coloneqq\can(G[\beta(t)];\alpha(t))$.
We apply Miller's algorithm to get a subcoset that respects $S$
and compute $\psi\Delta\coloneqq\can((\beta(t),\{S\});\psi'\Delta')$.

Compute a partial linear order on the graphs~$G_i$ by
comparing the resulting adjacency matrices of $G_i^{\tau_i}$ lexicographically
(Notice that the order does no depend on the representative $\tau_i$).
Define~$\chi\colon \{S_1,\ldots,S_\ell\} \rightarrow \{1,\ldots,h\}$ according to this order where~$h$ is the number of equivalence classes.
We refine $\chi$ such that $\chi(S_i)=\chi(S_j)$ implies
$\Theta_i=\Theta_j$.

If all adhesion sets are different,
we compute a canonical labeling coset $\Lambda$
of the set~$\beta(t)$ by
applying \cref{theo:big} to the
input $X=((\beta(t), \{S_i\}, \{\tau_i\Theta_i|_{S_i}\},\chi);\psi\Delta)$.
Otherwise, if all adhesion sets are equal, we
compute $\Lambda$ by applying the algorithm of
\cref{theo:small} to the
input~$X=((\beta(t), \{\tau_i\Theta_i|_{\beta(t)}\},\chi);\psi\Delta)$.
Notice that Property \ref{algo:cl1} allows us to restrict
$\tau_i\Theta_i$ to $S_i$.

We describe a canonical labeling coset of~$G$.

For each~$\lambda\in\Lambda$ we define extensions to~$V(G)$ as follows. 
For each~$i$ we define~$\tau'_{i}\Theta'_{i}$
to be the subcoset of those maps~$\rho_i\in
\tau_{i}\Theta_{i}$ for which~$\lambda^{-1}\rho_i$ is minimal.
(In fact, we have to restrict these maps to fit together
$(\lambda^{-1}|_{S_i^\lambda})(\rho_i|_{S_i})$).
In the case that all adhesions are
different, we
define~$\hat\Lambda_{\lambda}$ to be the set of labelings~$\hat\lambda$ of~$V(G)$
for which there are $\rho_1 \in
\tau'_{1}\Theta'_{1},\ldots,\rho_\ell\in \tau'_{\ell}\Theta'_{\ell}$
such that
\begin{enumerate}
  \item $\hat\lambda|_{\beta(t)}=\lambda$,
  \item $v_i^{\hat\lambda}<v_j^{\hat\lambda}$, if $v_i\in Z_i,v_j\in Z_j$ and
  $S_i^\lambda\prec S_j^\lambda$
  (the ordering $\prec$ of subsets of~$\NN$ as defined in the preliminaries)
  and,
  \item $v^{\hat\lambda}<u^{\hat\lambda}$, if $v,u\in Z_i$ and
  $v^{\rho_i}<u^{\rho_i}$.
\end{enumerate}
In the case that all adhesion sets are equal, we replace (2) by:
$v_i^{\hat\lambda}<v_j^{\hat\lambda}$, if $v_i\in Z_i,v_j\in Z_j$ and
$\lambda^{-1} \tau_i\Theta_i\prec_\chi \lambda^{-1} \tau_j\Theta_j$.
The former case implicitly implies $Z_i^{\hat\lambda}=Z^{\rho_i}+m$
for some $\rho_i\in\tau_i'\Theta_i'$ and some integer $m$.
However, in the latter case we additionally state this as an explicit condition. 

We define $\hat\Lambda$ as the smallest
coset containing~$\bigcup_{\lambda\in \Lambda}
\hat\Lambda_{\lambda}$.
Let $\mu\in\hat\Lambda$.
We define a map $\kappa:S^\mu\to\{1,\ldots,|S|\}$
such that $v^\kappa<u^\kappa$, if and only if
$v^\mu<u^\mu$. This is an isomorphism-invariant
renaming of the vertices such that  $S^{\mu\kappa}=\{1,\ldots,|S|\}$.
Return $\hat\Lambda\kappa$.
\end{cs}

\begin{canonicallabeling}
Property \ref{algo:cl1} is clearly satisfied.
We claim that for a partially coset-labeled decomposition $(T,\beta,\alpha)$
the output of the algorithm is a non-empty coset~$\CA(I)$ that satisfies
\ref{algo:cl2} and \ref{algo:cl3}.
The correctness of the first case is easy to see.
The returned set is non-empty because every graph of tree width
at most~$k$ is $k$-degenerate and therefore each subgraph has a
vertex of degree at most $k$.

We consider the second case.
 
\ref{algo:cl2} We consider the result of applying the algorithm
to~$G^\tau$ for any $\tau\in\lab(V(G))$.
So the initial input of the algorithm
is $(T,\beta)^\tau$ instead of $(T,\beta)$.
Let $\tau|_i:=\tau|_{V(G_i)}$.
The children of the root of~$T^\tau$ has the
decompositions~$(T_i,\beta_i,\alpha_i)^{\tau|_i}$ with the corresponding
graphs~$G_i^{\tau|_i}$ instead.
By induction we assume that \ref{algo:cl2} holds
for subgraphs of $G$.
Property \ref{algo:cl2} implies that even
for an arbitrary bijection $\tau|_i$ (which image is not necessary
a start segment of $\NN$)
it holds that $(\tau_i\Theta_i)^{\tau|_i}
=\CA(G_i^{\tau|_i},S_i^{\tau|_i},(T_i,\beta_i,\alpha)^{\tau|_i})$.
Let $\tau|_0:=\tau|_{\beta(t)}$.
We get $\alpha(t)^{\tau|_0}$ instead of $\alpha(t)$
and get $(\psi\Delta)^{\tau|_0}$ instead of $\psi\Delta$.
Exemplarily we
explain the case when all adhesion sets are different.
In this case we get $\tau|_0^{-1}\Lambda$ instead of $\Lambda$.
Since
$\lambda^{-1}\rho_i=(\tau|_0^{-1}\lambda)^{-1}\tau|_i^{-1}\rho_i$
we chose $\tau|_i^{-1}\rho_i$ as a minimal element instead
of $\rho_i$ and this leads to $\tau|_i^{-1}\tau_i'\Theta_i'$
instead of $\tau_i'\Theta_i'$.
Because of (CL2) of the canonical labelings
of a hypergraph with coset-labeled hyperedges
we get $\tau|_0^{-1}\Lambda$ instead of $\Lambda$.
Since $v^{\hat\lambda}=v^{\tau\tau^{-1}\hat\lambda}$
for all $v\in G$, we get $\tau^{-1}\hat\Lambda_\lambda$
and finally $\tau^{-1}\bigcup_{\lambda\in\Lambda}
\hat\Lambda_\lambda$.
Observe that the smallest coset containing
$\tau^{-1}\bigcup_{\lambda\in\Lambda}
\hat\Lambda_\lambda$ is actually
$\tau^{-1}\hat\Lambda$ which shows the correctness.

\ref{algo:cl3}
The first inclusion $\aut(G,S,(T,\beta,\alpha))\pi \subseteq \CA(I)$ follows
from \ref{algo:cl2}.
It remains to prove $\CA(I)\subseteq
\aut(G)\pi$.
Assume we have $\hat\lambda,\hat\lambda'\in\bigcup_{\lambda\in \Lambda}
\hat\Lambda_{\lambda}$ and $\hat\lambda|_{\beta(t)}=\lambda$ and
$\hat\lambda'|_{\beta(t)}=\lambda'$.
We have to show $\hat\lambda\hat\lambda'^{-1}$ is an automorphism of $G$.
Since both
$\lambda,\lambda'\in\Lambda\leq\can(G[\beta(t)];\alpha(t))$
we conclude that $\lambda\lambda'^{-1}$ preserves the edge relation
$E(\beta(t))$.
With the same argument $\lambda\lambda'^{-1}$ preserves $\CS$.
Exemplarily, we explain the case when all adhesion sets are different.
We have $I^{\hat\lambda}=I^{\hat\lambda'}$
and therefore $\lambda\lambda'^{-1}$
preserves $\{S_i\},\{\tau_i\Theta_i|_{S_i}\}$ and $\chi$.
So assume that $\lambda\lambda'^{-1}$
maps $S_i$ to $S_j$.
It remains to show that $\hat\lambda\hat\lambda'^{-1}$
is an isomorphism from $G_i$ to $G_j$.
By construction we have
$\hat\lambda|_{Z_i}=\rho_i+m$ and $\hat\lambda'|_{Z_j}=\rho'_j+m'$ for
$\rho_i\in\tau_i'\Theta_i',\rho_j\in\tau_j'\Theta_j'$.
In fact it must be the case that $m=m'$,
because the labeling of $\hat\lambda,\hat\lambda'$
is ordered according to the $S_i^\lambda=S_j^{\lambda'}$
as required in (2).
Since $m=m'$ we conclude
$\hat\lambda|_{Z_i}(\hat\lambda'|_{Z_j})^{-1}=\rho_i|_{Z_i}(\rho'_j|_{Z_j})^{-1}$.
Since
$\lambda\lambda'^{-1}$ preserves $\{\tau_i\Theta_i|_{S_i}\}$
we conclude
$\lambda^{-1}\tau_i\Theta_i|_{S_i}=\lambda'^{-1}\tau_j\Theta_j|_{S_j}$
and we get
$\lambda^{-1}\rho_i|_{S_i}=\lambda'^{-1}\rho_j|_{S_j}$ by the minimality of the
elements in $\tau_i'\Theta_i'$ and $\tau_j'\Theta_j'$.
This gives us
$\hat\lambda|_{S_i}(\hat\lambda'|_{S_j})^{-1}
=\lambda|_{S_i}(\lambda'|_{S_j})^{-1}=\rho_i|_{S_i}(\rho'_j|_{S_j})^{-1}$.
We already noticed that
$\hat\lambda|_{Z_i}(\hat\lambda'|_{Z_j})^{-1}=\rho_i|_{Z_i}(\rho'_j|_{Z_j})^{-1}$
which leads to
$\hat\lambda|_{V_i}(\hat\lambda'|_{V_j})^{-1}=\rho_i|_{V_i}(\rho'_j|_{V_j})^{-1}$.
Since $\chi$ is preserved, we know that
that $G_i^{\rho_i}=G_j^{\rho_j}$
and therefore $G_i^{\hat\lambda\hat\lambda'^{-1}}=G_j$
which means that $\hat\lambda\hat\lambda'^{-1}$
is an isomorphism.
\end{canonicallabeling}

\begin{runningtime}
Due to dynamic programming the number of calls is polynomial in $|V(G)|$.
More precisely, for each bag of the clique separator decomposition,
we construct a coset-labeled tree decomposition and both have polynomial
many bags.
In the first case the construction
of the coset-labeled tree decomposition from \cref{theo:labelDecomposition}
takes time $2^{\CO(k^2\log k)}|V(G)|^{\CO(1)}$.
Consider the Case ``$\beta(t)$ is labeled'' in which all adhesion sets are
different.
By Property \ref{ii:2} of
the bag size of $\beta(t)$ is bounded by
$\funclarge$ and
by Property \ref{ii:4} the number of children is bounded by
$k\funclarge^2+2^k\funclarge$.
By Property \ref{ii:4} we know that $\cw(\alpha(t))\leq k+1$ and also
$\cw(\tau_i\Theta_i|_{S_i})\leq k+1$ since $\tau_i\Theta_i|_{S_i}$ is a
subcoset of $\alpha(t_i)\kappa$.
Therefore the canonization algorithm for the hypergraph with coset-labeled hyperedges
in \cref{theo:big} runs in
time
${(\funclarge(k\funclarge^2+2^k\funclarge))}^{
\CO(\max\{\cw(\alpha(t)),\cw(\alpha(t_i))\})}
\subseteq 2^{\CO(k^2\log k)}$.
Next, consider the second case in which all adhesion sets are equal.
In this case, we have that the bag size of $\beta(t)$ is bounded by
$\funcmedium$ due to Property \ref{ii:7} of our coset-labeled tree decomposition from
\cref{theo:labelDecomposition}.
Therefore the
canonization algorithm for the multiple colored cosets
in \cref{theo:small} runs in
time
$\funcmedium! |V(G)|^{\CO(1)}\subseteq 2^{\CO(k^2\log k)}|V(G)|^{\CO(1)}$.
\end{runningtime}
\end{proof}

\bibliographystyle{alpha}
\bibliography{references}

\begin{thebibliography}{LPPS17}

\bibitem[Bab16]{DBLP:conf/stoc/Babai16}
L{\'{a}}szl{\'{o}} Babai.
\newblock Graph isomorphism in quasipolynomial time [extended abstract].
\newblock In Daniel Wichs and Yishay Mansour, editors, {\em Proceedings of the
  48th Annual {ACM} {SIGACT} Symposium on Theory of Computing, {STOC} 2016,
  Cambridge, MA, USA, June 18-21, 2016}, pages 684--697. {ACM}, 2016.

\bibitem[BKL83]{DBLP:conf/focs/BabaiKL83}
L{\'{a}}szl{\'{o}} Babai, William~M. Kantor, and Eugene~M. Luks.
\newblock Computational complexity and the classification of finite simple
  groups.
\newblock In {\em 24th Annual Symposium on Foundations of Computer Science,
  Tucson, Arizona, USA, 7-9 November 1983}, pages 162--171. {IEEE} Computer
  Society, 1983.

\bibitem[BL83]{babai1983canonical}
L{\'a}szl{\'o} Babai and Eugene~M Luks.
\newblock Canonical labeling of graphs.
\newblock In {\em Proceedings of the fifteenth annual ACM symposium on Theory
  of computing}, pages 171--183. ACM, 1983.

\bibitem[Bod90]{DBLP:journals/jal/Bodlaender90}
Hans~L. Bodlaender.
\newblock Polynomial algorithms for graph isomorphism and chromatic index on
  partial k-trees.
\newblock {\em J. Algorithms}, 11(4):631--643, 1990.

\bibitem[ES16]{DBLP:conf/stacs/ElberfeldS16}
Michael Elberfeld and Pascal Schweitzer.
\newblock Canonizing graphs of bounded tree width in logspace.
\newblock In {\em {STACS}}, volume~47 of {\em LIPIcs}, pages 32:1--32:14.
  Schloss Dagstuhl - Leibniz-Zentrum fuer Informatik, 2016.

\bibitem[GM15]{gromar15}
M.~Grohe and D.~Marx.
\newblock Structure theorem and isomorphism test for graphs with excluded
  topological subgraphs.
\newblock {\em sicomp}, 44(1):114--159, 2015.

\bibitem[GNS18]{bounded-degree-small-diameter}
Martin Grohe, Daniel Neuen, and Pascal Schweitzer.
\newblock Towards faster isomorphism tests for bounded degree graphs.
\newblock {\em CoRR}, abs/1802.04659, 2018.

\bibitem[GS15]{DBLP:conf/focs/GroheS15}
Martin Grohe and Pascal Schweitzer.
\newblock Isomorphism testing for graphs of bounded rank width.
\newblock In Venkatesan Guruswami, editor, {\em {IEEE} 56th Annual Symposium on
  Foundations of Computer Science, {FOCS} 2015, Berkeley, CA, USA, 17-20
  October, 2015}, pages 1010--1029. {IEEE} Computer Society, 2015.

\bibitem[Lei93]{DBLP:journals/dm/Leimer93}
Hanns{-}Georg Leimer.
\newblock Optimal decomposition by clique separators.
\newblock {\em Discrete Mathematics}, 113(1-3):99--123, 1993.

\bibitem[LPPS17]{lokshtanov2017fixed}
Daniel Lokshtanov, Marcin Pilipczuk, Micha{\l} Pilipczuk, and Saket Saurabh.
\newblock Fixed-parameter tractable canonization and isomorphism test for
  graphs of bounded treewidth.
\newblock {\em SIAM Journal on Computing}, 46(1):161--189, 2017.

\bibitem[LS99]{liebeck1999simple}
Martin Liebeck and Aner Shalev.
\newblock Simple groups, permutation groups, and probability.
\newblock {\em Journal of the American Mathematical Society}, 12(2):497--520,
  1999.

\bibitem[Luk82]{luks1982isomorphism}
Eugene~M Luks.
\newblock Isomorphism of graphs of bounded valence can be tested in polynomial
  time.
\newblock {\em Journal of computer and system sciences}, 25(1):42--65, 1982.

\bibitem[Mat79]{DBLP:journals/ipl/Mathon79}
Rudolf Mathon.
\newblock A note on the graph isomorphism counting problem.
\newblock {\em Inf. Process. Lett.}, 8(3):131--132, 1979.

\bibitem[Mil83]{miller1983isomorphism}
Gary Miller.
\newblock Isomorphism testing and canonical forms for k-contractable graphs (a
  generalization of bounded valence and bounded genus).
\newblock In {\em Foundations of Computation Theory}, pages 310--327. Springer,
  1983.

\bibitem[Neu16]{Neuen16}
Daniel Neuen.
\newblock Graph isomorphism for unit square graphs.
\newblock In Piotr Sankowski and Christos~D. Zaroliagis, editors, {\em 24th
  Annual European Symposium on Algorithms, {ESA} 2016, August 22-24, 2016,
  Aarhus, Denmark}, volume~57 of {\em LIPIcs}, pages 70:1--70:17. Schloss
  Dagstuhl - Leibniz-Zentrum fuer Informatik, 2016.

\bibitem[OS14]{DBLP:conf/swat/OtachiS14}
Yota Otachi and Pascal Schweitzer.
\newblock Reduction techniques for graph isomorphism in the context of width
  parameters.
\newblock In {\em {SWAT}}, volume 8503 of {\em Lecture Notes in Computer
  Science}, pages 368--379. Springer, 2014.

\bibitem[Pon91]{Ponomarenko}
Ilia~N.\ Ponomarenko.
\newblock The isomorphism problem for classes of graphs closed under
  contraction.
\newblock {\em Journal of Mathematical Sciences}, 55(2):1621--1643, 1991.

\end{thebibliography}
\addcontentsline{toc}{chapter}{Bibliography}

\end{document}